\documentclass[11pt]{article}

\usepackage{amsmath}
\usepackage{amsfonts}
\usepackage{amsthm}
\usepackage[letterpaper, margin=1in]{geometry}
\usepackage[utf8]{inputenc}
\usepackage{xkeyval}
\usepackage[ruled, vlined]{algorithm2e}
\usepackage{complexity}

\usepackage{subfig}

\usepackage{tikz}
\usetikzlibrary{arrows,shapes, positioning}

\usepackage{float} 
\floatstyle{boxed} 
\restylefloat{figure}

\newcommand{\reals}{\mathbb{R}}

\newtheorem{lemma}{Lemma}
\newtheorem{theorem}{Theorem}
\newtheorem{definition}{Definition}
\newtheorem{corollary}{Corollary}

\newlang{\SETH}{\mathrm{SETH}}
\newlang{\NSETH}{\mathrm{NSETH}}
\newlang{\cnfsat}{\textsc{cnfsat}}
\newcommand{\ksat}[1][k]{\lang{\textsc{$#1$-sat}}}

\DeclareMathOperator{\insrt}{insert}
\DeclareMathOperator{\delete}{delete}
\DeclareMathOperator{\idx}{index}
\DeclareMathOperator{\deleteMin}{delete\_min}
\DeclareMathOperator{\rank}{rank}
\DeclareMathOperator{\head}{head}
\DeclareMathOperator{\nxt}{next}
\DeclareMathOperator{\maj}{THR}
\DeclareMathOperator{\val}{val}
\DeclareMathOperator{\sgn}{sgn}
\DeclareMathOperator{\Query}{Query}

\author{Marvin K\"unnemann\thanks{Part of this research was conducted while being affiliated with Max Planck Institute for Informatics, Saarbr\"ucken, Germany and visiting the Simons Institute for the Theory of Computing, Berkeley.} \qquad Daniel Moeller\thanks{This research was partially supported by the Army Research Office grant number W911NF-15-1-0253.} \footnotemark[3] \qquad Ramamohan Paturi\footnotemark[2] \thanks{This research is supported by the Simons Foundation and was partially conducted while visiting the Simons Institute for the Theory of Computing, Berkeley. This research is supported by NSF
grant CCF-1213151 from the
Division  of Computing and Communication Foundations.
Any opinions,
findings and conclusions or
recommendations expressed in this material are those
of the authors and do
not necessarily reflect the
views of the National Science Foundation. The
conference version is available at Springer via
http://dx.doi.org/10.1007/978-3-319-34171-2\_21} \qquad Stefan Schneider\footnotemark[3]
  \\ \\
  Department of Computer Science and Engineering\\UC, San Diego\\La Jolla, CA 92093}

\title{Subquadratic Algorithms for Succinct Stable Matching}
\begin{document}
\maketitle

\begin{abstract}
  We consider the stable matching problem when the preference lists
  are not given explicitly but are represented in a succinct way and
  ask whether the problem becomes computationally easier and
  investigate other implications. We give subquadratic algorithms for
  finding a stable matching in special cases of natural succinct
  representations of the problem, the $d$-attribute, $d$-list,
  geometric, and single-peaked models. We also present algorithms for
  verifying a stable matching in the same models. We further show that
  for $d = \omega(\log n)$ both finding and verifying a stable
  matching in the $d$-attribute and $d$-dimensional geometric models
  requires quadratic time assuming the Strong Exponential Time
  Hypothesis. This suggests that these succinct models are not
  significantly simpler computationally than the general case for
  sufficiently large $d$.
\end{abstract}

\section{Introduction}
The stable matching problem has applications that vary from
coordinating buyers and sellers to assigning students to public
schools and residents to hospitals
\cite{Gusfield1989,knuth1997stable,Roth1990}. Gale and Shapley \cite{GaleShapley1962}
proposed a quadratic time \emph{deferred acceptance} algorithm for
this problem which has helped clear matching markets in many real-world settings. For arbitrary preferences, the deferred acceptance algorithm is optimal and even verifying that a given matching is stable requires quadratic
time
\cite{ng1990lower,segal2007communication,gonczarowski2015stable}. This is reasonable since representing all participants' preferences requires quadratic space. However,
in many applications the preferences are not arbitrary and can
have more structure. For example, top doctors are likely to be
universally desired by residency programs and students typically seek
highly ranked schools. In these cases
participants can represent their preferences succinctly. It is
natural to ask whether the same quadratic time bounds apply with compact and
structured preference models that have subquadratic representations. This will provide a more nuanced understanding of where the complexity lies: Is stable matching inherently complex, or is the complexity merely a result of the large variety of possible preferences? To this end, we examine several restricted
preference models with a particular focus on two originally proposed by Bhatnagar et
al. \cite{bhatnagar2008sampling}, the $d$-attribute and $d$-list
models. Using a wide range of techniques we provide algorithms and conditional hardness results for several settings of these models.

In the \emph{$d$-attribute} model, we assume that there are $d$ different
attributes (e.g.~income, height, sense of humor, etc.) with a fixed, possibly objective, ranking of the men for each attribute. Each woman's
preference list is based on a linear combination of the attributes of
the men, where each woman can have different weights for each
attribute. Some women may care more about, say, height whereas others
care more about sense of humor. Men's preferences are defined
analogously. This model is applicable in large settings, such as
online dating systems, where participants lack the resources to form
an opinion of every other participant. Instead the system can rank the
members of each gender according to the $d$ attributes and each
participant simply needs to provide personalized weights for the
attributes. The combination of attribute values and weights implicitly
represents the entire preference
matrix. Bogomolnaia and Laslier \cite{bogomolnaia2007euclidean} show that representing all
possible $n\times n$ preference matrices requires $n-1$
attributes. Therefore it is reasonable to expect that when $d \ll n-1$,
we could beat the worst case quadratic lower bounds for the general
stable matching problem.

In the \emph{$d$-list} model, we assume that there are $d$ different
rankings of the men. Each women selects one of the $d$ lists as her
preference list. Similarly, each man chooses one of $d$ lists of women
as his preference list. This model captures the setting where members
of one group (i.e.~student athletes, sorority members, engineering
majors) may all have identical preference lists. Mathematically, this model is actually a special case of the
$d$-attribute model where each participant places a positive weight on
exactly one attribute. However, its motivation is distinct and we can
achieve improved results for this model.

Chebolu et al.~prove that approximately counting stable matchings in
the $d$-attribute model for $d \geq 3$ is as hard as the general case
\cite{chebolu2010complexity}. Bhatnagar et al. showed that sampling
stable matchings using random walks can take exponential time even for
a small number of attributes or lists but left it as an open question
whether subquadratic algorithms exist for these models
\cite{bhatnagar2008sampling}.

We show that faster algorithms exist for finding a stable matching in some special cases of these
models. In particular, we provide subquadratic algorithms for the
$d$-attribute model, where all values and weights are from a small
set, and the one-sided $d$-attribute model, where one side of the
market has only one attribute. These results show we can achieve
meaningful improvement over the general setting for some restricted
preferences.

While we only provide subquadratic algorithms to find
stable matchings in special cases of the attribute model, we have
stronger results concerning verification of stable matchings. We
demonstrate optimal subquadratic stability testing algorithms for the
$d$-list and boolean $d$-attribute settings as well as a subquadratic
algorithm for the general $d$-attribute model with constant $d$. These
algorithms provide a clear distinction between the attribute model and
the general setting. Moreover, these results raise the question of
whether verifying and finding a stable matching are equally hard
problems for these restricted models, as both require quadratic time in the general case.

Additionally, we show that the stable matching problem in the $d$-attribute
model for $d = \omega(\log n)$ cannot be solved in subquadratic time
under the Strong Exponential Time Hypothesis ($\SETH$)
\cite{ImpagliazzoPaturi2001, ImpagliazzoPaturiZane2001}. We show
$\SETH$-hardness for both finding and verifying a stable matching and for checking if a given pair is in any or all stable matchings,
even when the weights and attributes are boolean. This adds the stable
matching problem to a growing list of $\SETH$-hard problems, including
Fr\'{e}chet distance \cite{Bringmann2014}, edit distance
\cite{BackursIndyk2015}, string matching
\cite{AbboudBackursWilliams2015}, $k$-dominating set
\cite{PatrascuWilliams2010}, orthogonal vectors \cite{Williams2004},
and vector domination \cite{ImpagliazzoPaturiSchneider2013}. Thus the
quadratic time hardness of the stable matching problem in the general
case extends to the more restricted and succinct $d$-attribute
model. This limits the space of models where we can hope to find
subquadratic algorithms.

We further present several results in related succinct preference areas. Single-peaked preferences are commonly used to model preferences in social choice theory because of their simplicity and because they often approximate reality. Essentially, single-peaked preferences require that everyone agree on a common spectrum along which all alternatives can be ranked. However, each individual may have a different ideal choice and prefers the ``closest'' alternatives. A typical example is the political spectrum where candidates fall somewhere between liberal and conservative. In this setting, voters tend to prefer the candidates that are closer to their own ideals. As explained below, these preferences can be succinctly represented. Bartholdi and Trick \cite{BartholdiTrick1986} present a subquadratic time algorithm for stable roommates (and stable matching) with narcissistic, single-peaked preferences. In the narcissistic case, the participants are located at their own ideals. This makes sense in some applications but is not always realistic. We provide a subquadratic algorithm to verify if a given matching is stable in the general single-peaked preference model. Chung uses a slightly different model of single-peaked preferences where a stable roommate matching always exists \cite{Chung2000}. In this model the participants would rather be unmatched than matched with someone further away from their ideal than they are themselves, leading to incomplete preference lists.

We extend our algorithms and lower bounds for the attribute model
to the geometric model where preference orders are formed according to
euclidean distances among a set of points in multi-dimensional space.
Arkin et al.~\cite{ArkinBaeEfrat2009} derive a subquadratic algorithm
for stable roommates with narcissistic geometric preferences in
constant dimensions. Our algorithms do not require the preferences to be narcissistic.

It is worth noting that all of our verification and hardness results apply to the stable roommates problem as well. This problem is identical to stable matching except we remove the bipartite distinction between the participants \cite{Gusfield1989}. Unlike with bipartite stable matching, there need not always exist a stable roommate matching. However Irving discovered an algorithm that produces a stable matching or identifies that none exists in quadratic time \cite{Irving1985}. Since finding a stable roommate matching is strictly harder than finding a stable matching, this is also optimal. Likewise, verification is equally hard for both stable roommates and stable matching, as we can simply duplicate every participant and treat the roommate matching as bipartite. Therefore, our results show that verification can be done more efficiently for the stable roommates problem when the preferences are succinct.

Finally, we address the issue of strategic behavior in these restricted models. It is often preferable for a market-clearing mechanism to incentivize truthful behavior from the participants so that the outcome faithfully captures the optimal solution. Particularly in matching markets, this objective complements the desire for a stable matching where participants have incentives to cooperate with the outcome. Roth \cite{Roth1982} showed that there is no strategy proof mechanism to find a stable matching in the general preferences setting. Additionally, if a mechanism outputs the man-optimal stable matching, the women can manipulate it to obtain the woman-optimal solution by truncating their preference lists \cite{Roth1982,GaleSotomayor1985}. Even if the women are required to rank all men, they can still achieve more preferable outcomes in some instances \cite{TeoSethuramanTan2001,KobayashiMatsui2010}. However, in the $d$-attribute, $d$-list, single-peaked, and geometric preference models, there are considerably fewer degrees of freedom for preference misrepresentation. Nevertheless, we show that there is still no strategy proof mechanism to find a stable matching for any of these models with $d\geq 2$ and non-narcissistic preferences.

Dabney and Dean \cite{dabney2010adaptive} study an alternative
succinct preference representation where there is a canonical
preference list for each side and individual deviations from this list
are specified separately. They provide an adaptive $O(n+k)$ time
algorithm for the special one-sided case, where $k$ is the number of
deviations.

\section{Summary of Results}
Section \ref{sec:findsmall} gives an $O(C^{2d} n (d + \log n))$ time algorithm for
finding a stable matching in the $d$-attribute model if both the attributes
and weights are from a set of size at most $C$. This gives a strongly
subquadratic algorithm (i.e.~$O(n^{2-\varepsilon})$ for $\varepsilon
>0$) if $d < \frac1{2\log C} \log n$. 

Section \ref{sec:onesided} considers an asymmetric case, where one
side of the matching market has $d$ attributes, while the other side
has a single attribute. We allow both the weights and attributes to be
arbitrary real values. Our algorithm for finding a stable matching in
this model has time complexity $\tilde{O}(n^{2-1/\lfloor d/2
  \rfloor})$, which is strongly subquadratic for constant $d$. 

In Section \ref{sec:verifyreal} we consider the problem of verifying
that a given matching is stable in the $d$-attribute model with real
attributes and weights. The time complexity of our algorithm is
$\tilde{O}(n^{2-1/2d})$, which is again strongly subquadratic for
constant $d$. 

Section \ref{sec:lists} gives an $O(dn)$ time algorithm for verifying a stable
matching in the $d$-list model. This is
linear in its input size and is therefore optimal.

In Section \ref{sec:verifyboolean} we give a randomized $\tilde{O}(n^{2-1/O(c
  \log^2(c))})$ time algorithm for $d = c \log n$ for verifying a
stable matching in the $d$-attribute model when both the weights and
attributes are boolean. This algorithm is strongly subquadratic for $d
= O(\log n)$.

In Section \ref{sec:hardness} we give a conditional lower bound for
the three problems of finding and verifying a stable matching in the
$d$-attribute model as well as the stable pair problem. We show that
there is no strongly subquadratic algorithm for any of these problems
when $d = \omega(\log n)$ assuming the Strong Exponential Time
Hypothesis. For the stable pair problem we give further evidence that
even nondeterminism does not give a subquadratic algorithm.

Finally in Section \ref{sec:otherpref} we consider the related
preference models of single-peaked and geometric preferences. We
extend our algorithms and lower bounds for the attribute model to the
geometric model and give an $O(n \log n)$ algorithm for verifying a
stable matching with single-peaked preferences.

\section{Preliminaries}
\label{sec:preliminaries}

A \emph{matching market} consists of a set of men $M$ and a set of
women $W$ with $|M| = |W| = n$. We further have a permutation of $W$
for every $m \in M$, and a permutation of $M$ for every $w \in W$,
called \emph{preference lists}. Note that representing a general
matching market requires size $\Omega(n^2)$.

For a perfect bipartite matching $\mu$, a \emph{blocking pair} with
respect to $\mu$ is a pair $(m,w) \not\in \mu$ where $m \in M$ and $w
\in W$, such that $w$ appears before $\mu(m)$ in $m$'s preference list
and $m$ appears before $\mu(w)$ in $w$'s preference list. A perfect bipartite
matching is called \emph{stable} if there are no blocking pairs. In settings where ties in the preference lists are possible, we consider weakly stable matchings where $(m,w)$ is a blocking pair if and only if both strictly prefer each other to their partner. For simplicity, we assume all preference lists are complete though our results trivially extend to cases with incomplete lists.

Gale's and Shapley's deferred acceptance algorithm
\cite{GaleShapley1962} works as follows. While there is an
unmatched man $m$, have $m$ \emph{propose} to his most preferred
woman who has not already rejected him. A woman accepts a proposal if
she is unmatched or if she prefers the proposing man to her current
partner, leaving her current partner unmatched. Otherwise, she rejects
the proposal. This process finds a stable matching in time $O(n^2)$.

A matching market in the $d$-attribute model consists of $n$ men and
$n$ women as before. A participant $p$ has attributes $A_i(p)$ for $1
\leq i \leq d$ and weights $\alpha_i(p)$ for $1 \leq i \leq d$. For a
man $m$ and woman $w$, $m$'s \emph{value} of $w$ is given by
$\val_m(w) = \langle \alpha(m), A(w) \rangle = \sum_{i=1}^d
\alpha_i(m) A_i(w)$. $m$ ranks the women in decreasing order of
value. Symmetrically, $w$'s value of $m$ is $\val_w(m) = \sum_{i=1}^d
\alpha_i(w)A_i(m)$. Note that representing a matching market in the
$d$-attribute model requires size $O(dn)$. Unless otherwise specified,
both attributes and weights can be negative.

A matching market in the $d$-list model is a matching market where
both sides have at most $d$ distinct preference lists. Describing a
matching market in this model requires $O(dn)$ numbers.

Throughout the paper, we use $\tilde{O}$ to suppress polylogarithmic
factors in the time complexity.

\section{Finding Stable Matchings}

\subsection{Small Set of Attributes and Weights}
\label{sec:findsmall}

We first present a stable matching algorithm for the $d$-attribute
model when the attribute and weight values are limited to a set of constant size. In particular, we assume that the number of
possible values for each attribute and weight for all participants is
bounded by a constant $C$.

\begin{algorithm}[H]
Group the women into sets $S_i$ with a set for each of the $C' = O(C^{2d})$ types of women. ($O(C^d)$ possible attribute values and $O(C^d)$ possible weight vectors.)\\

Associate an empty min-heap $h_i$ with each set $S_i$.\\
\For{each man $m$}{
	Create $m$'s preference list of sets $S_i$.\\
	$\idx (m)\gets 1$
}
\While{there is a man $m$ who is not in any heap}{
	Let $S_i$ be the $\idx(m)$ set on $m$'s list.\\
	\If{$|h_i| < |S_i|$}{
		$h_i.\insrt(m)$
	}
	\Else{
		\If{$\val_{S_i}(m) > \val_{S_i}(h_i.min)$}{
			$h_i.\deleteMin()$\\
			$h_i.\insrt(m)$\\
		}
	}
	$\idx(m) \gets \idx(m) + 1$\\
}
\For{$i = 1$ to $C'$}{
$\mu\gets \mu\bigcup$ Arbitrarily pair women in $S_i$ with men in $h_i$.\\
}
\Return{$\mu$}
\caption{Small Constant Attributes and Weights}
\label{alg:smallstablematching}
\end{algorithm}

\begin{theorem}\label{thm:smallstablematching}
  There is an algorithm to find a stable matching in the
  $d$-attribute model with at most a constant $C$ distinct attribute and weight 
  values in time $O(C^{2d}n(d + \log n))$.
\end{theorem}

\begin{proof}
Consider Algorithm \ref{alg:smallstablematching}. First observe that each man is indifferent between the women in a given set $S_i$ because each woman has identical attribute values. Moreover, the women in a set $S_i$ share the same ranking of the men, since they have identical weight vectors. Therefore, since we are looking for a stable matching, we can treat each set of women $S_i$ as an individual entity in a many to one matching where the capacity for each $S_i$ is the number of women it contains.

With these observations, the stability follows directly from the stability of the standard deferred acceptance algorithm for many-one stable matching. Indeed, each man proposes to the sets of women in the order of his preferences and each set of women tentatively accepts the best proposals, holding onto no more than the available capacity.

The grouping of the women requires $O(C^{2d} + dn)$ time to initialize the groups and place each woman in the appropriate group. Creating the men's preference lists requires $O(dC^{2d}n)$ time to evaluate and sort the groups of women for every man. The while loop requires $O(C^{2d}n(d + \log n))$ time since each man will propose to at most $C^{2d}$ sets of women and each proposal requires $O(d + \log n)$ time to evaluate and update the heap. This results in an overall running time of $O(C^{2d}n(d+\log n))$.
\end{proof}

As long as $d < \frac{1}{2\log C} \log n$, the time complexity in Theorem \ref{thm:smallstablematching} will be subquadratic. It is worth noting that the algorithm and proof actually do not rely on any restriction of the men's attribute and weight values. Thus, this result holds whenever one side's attributes and weight values come from a set of constant size.

\subsection{One-Sided Real Attributes}
\label{sec:onesided}
In this section we consider a one-sided attribute model with real
attributes and weights. In this model, women have $d$ attributes and
men have $d$ weights, and the preference list of a man is given by the
weighted sum of the women's attributes as in the two-sided attribute
model. On the other hand there is only one attribute for the men. The
women's preferences are thus determined by whether they have a
positive or negative weight on this attribute. For simplicity, we
first assume that all women have a positive weight on the men's
attribute and show a subquadratic algorithm for this case. Then we
extend it to allow for negative weights.

To find a stable matching when the women have a global preference list
over the men, we use a greedy approach: process the men from the most
preferred to the least preferred and match each man with the highest
unmatched woman in his preference list. This general technique is not specific to the attribute model but
actually works for any market where one side has a single global
preference list. (e.g.~\cite{dabney2010adaptive} uses a similar
approach for their algorithm.) The complexity lies in repeatedly
finding which of the available women is most preferred by the current
top man.

This leads us to the following algorithm: for every woman $w$ consider
a point with $A(w)$ as its coordinates and organize the set of
points into a data structure. Then, for the men in order of preference, query the set of
points against a direction vector consisting of the man's weight and
find the point with the largest distance along this direction. Remove that point and repeat.

The problem of finding a maximal point along a direction is typically
considered in its dual setting, where it is called the \emph{ray
  shooting problem}. In the ray shooting problem we are given $n$
hyperplanes and must maintain a data structure to answer queries. Each query consists of a vertical ray and the data structure returns the first hyperplane hit by that ray.

The relevant results are in Lemma \ref{lem:ray} which follows from several papers for different values of $d$. For an overview of the ray shooting problem and related range query problems, see \cite{AgarwalErickson1999}.

\begin{lemma}[\cite{HershbergerSuri1995, DobkinKirkpatrick1985,
    AgarwalArgeErickson1998, MatousekSchwarzkopf1992}]
  \label{lem:ray}
  Given an $n$ point set in $\reals^d$ for $d \geq 2$, there is a data
  structure for ray shooting queries with preprocessing time
  $\tilde{O}(n)$ and query time $\tilde{O}(n^{1-1/\lfloor d/2
    \rfloor})$. The structure supports deletions with amortized update
  time $\tilde{O}(1)$.
\end{lemma}

For $d=1$, queries can trivially be answered in constant time. We use this data structure to provide an algorithm when there is a global list for one side of the market.

\begin{lemma}\label{lem:onesided}
  For $d \geq 2$ there is an algorithm to find a stable matching in the
  one-sided $d$-attribute model with real-valued attributes and weights in time
  $\tilde{O}(n^{2-1/\lfloor d/2 \rfloor})$ when there is a single preference 
  list for the other side of the market.
\end{lemma}
\begin{proof}
  For a man $m$, let $\dim(m)$ denote the index of the last non-zero
  weight, i.e.~$\alpha_{\dim(m)+1}(m) = \cdots = \alpha_d(m) = 0$. We
  assume $\dim(m) > 0$, as otherwise $m$ is indifferent among all
  women and we can pick any woman as $\mu(m)$. We assume without loss
  of generality $\alpha_{\dim(m)}(m) \in \{-1,1\}$.  For each $d'$
  such that $1 \leq d' \leq d$ we build a data structure consisting of
  $n$ hyperplanes in $\reals^{d'}$. For each woman $w$, consider the
  hyperplanes
  \begin{equation}
    \label{eq:hypdef}
    H_{d'}(w) = \left\{x_{d'} = \sum_{i=1}^{d'-1} A_i(w) x_i - A_{d'}(w)\right\}
  \end{equation}
  and for each $d'$ preprocess the set of all hyperplanes according to Lemma
  \ref{lem:ray}. Note that $H_{d'}(w)$ is the dual of the
  point $(A_1(w), \ldots, A_{d'}(w))$.

  For a man $m$ we can find his most preferred partner by querying the
  $\dim(m)$-dimensional data structure. Let $s =
  \alpha_{\dim(m)}(m)$. Consider a ray $r(m) \in \reals^{\dim(m)}$
  originating at
  
  \begin{equation}
    ({-\frac{\alpha_1(m)}{s}}, \ldots,
    {-\frac{\alpha_{\dim(m)-1}(m)}{s}}, -s\cdot\infty)
  \end{equation}
  
\noindent
in the direction $(0, \ldots, 0,s)$. If $\alpha_{\dim(m)} = 1$ we find the
lowest hyperplane intersecting the ray, and if $\alpha_{\dim(m)} = -1$ we find
the highest hyperplane. We claim that the first hyperplane $r(m)$ hits
corresponds to $m$'s most preferred woman. Let woman $w$ be preferred
over woman $w'$,
i.e.~$\val_m(w) = \sum_{i=1}^{\dim(m)} A_i(w) \alpha_i(m) \geq
\sum_{i=1}^{\dim(m)} A_i(w') \alpha_i(m) = \val_m(w')$.
Since the ray $r(m)$ is vertical in coordinate $x_{d'}$, it is
sufficient to evaluate the right-hand side of the definition in
equation \ref{eq:hypdef}. Indeed we have $\val_m(w) \geq \val_m(w')$
if and only if
  \begin{equation}
    \sum_{i=1}^{\dim(m)-1} {-A_i(w)} \frac{\alpha_i(m)}s -
  A_{\dim(m)}(w) \leq \sum_{i=1}^{\dim(m)-1} {-A_i(w')} \frac{\alpha_i(m)}s -
  A_{\dim(m)}(w')
  \end{equation}
  when $s = 1$ and
  \begin{equation}
    \sum_{i=1}^{\dim(m)-1} {-A_i(w)} \frac{\alpha_i(m)}s -
  A_{\dim(m)}(w) \geq \sum_{i=1}^{\dim(m)-1} {-A_i(w')} \frac{\alpha_i(m)}s -
  A_{\dim(m)}(w')
  \end{equation}
  when $s = -1$.

  Note that the query ray is dual to the set of hyperplanes with normal vector
  $(\alpha_1(m), \ldots, \alpha_d(m))$.

  Now we pick the highest man $m$ in the (global) preference list,
  consider the ray as above and find the first hyperplane
  $H_{\dim(m)}(w)$ hit by the ray. We then match the pair $(m,w)$,
  remove $H(w)$ from all data structures and repeat. Correctness
  follows from the correctness of the greedy approach when all women
  share the same preference list and the properties of the halfspaces
  proved above.

  The algorithm preprocesses $d$ data structures, then makes $n$ queries
  and $dn$ deletions. The time is dominated by the $n$ ray queries each requiring time $\tilde{O}(n^{1-1/\lfloor d/2 \rfloor})$. Thus the total time complexity is bounded by $\tilde{O}(n^{2-1/\lfloor d/2 \rfloor})$, as claimed.
\end{proof} 

\begin{algorithm}[H]
// For points in $P \in \reals^{d}$ we use the notation $(x_1,
\ldots, x_{d})$ to refer to its coordinates. \\

Input: matching $\mu$\\
\For{$d' = 1$ to $d$}{
  \For{each woman $w$}{
    $H(w) \gets \{x_d = \sum_{i=1}^{d-1} A_i(w) x_i - A_{d'}(w)\}$ \\
    $H_{d'} \gets H_{d'} \cup H(w)$
  }
  $H_{d'}$.preprocess()
}

\For{each man $m$ in order of preference} {
  $s \gets \alpha_{\dim(m)}(m)$ \\
  $r(m) \gets (-\frac{\alpha_1(m)}{s}, \ldots,
  -\frac{\alpha_{\dim(m)-1}(m)}{s}, \infty\cdot s) + t \cdot
  (0, \ldots, 0,-s)$ \\
  $H(w) \gets \Query(H_{\dim(m)}, r(m))$ \\
  $\mu \gets \mu \cup (m,w)$ \\
  
  \For{$d' = 1$ to $d$}{
    $H_{d'} \gets H_{d'} - H_{d'}(w)$
  }
}
\Return{$\mu$}
\caption{One-Sided Stable Matching}
\label{alg:onesided}
\end{algorithm}

Note that for $d=1$ there is a trivial linear time algorithm for the problem.

We use the following lemma to extend the above algorithm to account
for positive and negative weights for the women. It deals with
settings where the women choose one of two lists ($\sigma_1,\sigma_2$)
as their preference lists over the men while the men's preferences can
be arbitrary.

\begin{lemma}\label{lem:listflip}
  Suppose there are $k$ women who use $\sigma_1$. If the top $k$ men
  in $\sigma_1$ are in the bottom $k$ places in $\sigma_2$, then the
  women using $\sigma_1$ will only match with those men and the $n-k$
  women using $\sigma_2$ will only match with the other $n-k$ men in
  the woman-optimal stable matching.
\end{lemma}

\begin{proof}
  Consider the operation of the woman-proposing deferred acceptance
  algorithm for finding the woman-optimal stable matching. Suppose the
  lemma is false so that at some point a woman using $\sigma_1$
  proposed to one of the last $n-k$ men in $\sigma_1$. Let $w$ be the
  first such woman. $w$ must have been rejected by all of the top $k$,
  so at least one of those men received a proposal from a woman, $w'$,
  using $\sigma_2$. However, since the top $k$ men in $\sigma_1$ are
  the bottom $k$ men in $\sigma_2$, $w'$ must have been rejected by
  all of the top $n-k$ men in $\sigma_2$. But there are only $n-k$
  women using $\sigma_2$, so one of the top $n-k$ men in $\sigma_2$
  must have already received a proposal from a woman using
  $\sigma_1$. This is a contradiction because $w$ was the first woman
  using $\sigma_1$ to propose to one of the bottom $n-k$ men in
  $\sigma_1$ (which are the top $n-k$ men in $\sigma_2$).
\end{proof}

We can now prove the following theorem where negative values are
allowed for the women's weights.

\begin{theorem}\label{thm:onesided}
  For $d \geq 2$ there is an algorithm to find a stable matching in
  the one-sided $d$-attribute model with real-valued attributes and
  weights in time $\tilde{O}(n^{2-1/\lfloor d/2 \rfloor}).$
\end{theorem}

\begin{proof}
  Suppose there are $k$ women who have a positive weight on the men's
  attribute. Since the remaining $n-k$ women's preference list is the
  reverse, we can use Lemma \ref{lem:listflip} to split the problem
  into two subproblems. Namely, in the woman-optimal stable matching
  the $k$ women with a positive weight will match with the top $k$
  men, and the $n-k$ women with a negative weight will match with the
  bottom $n-k$ men. Now the women in each of these subproblems all
  have the same list. Therefore we can use Lemma \ref{lem:onesided} to
  solve each subproblem. Splitting the problem into subproblems can be
  done in time $O(n)$ so the running time follows immediately from
  Lemma \ref{lem:onesided}.
\end{proof}

\begin{table}[H]
  \caption{Two-list preferences where no participant receives their top choice in the stable matching}
  \label{tab:preferences}
  \begin{center}
    \begin{tabular}{c|c c c|c}
      $\sigma_1$ & $\sigma_2$ & & $\pi_1$ & $\pi_2$\\
      \hline
      $m_1$ & $m_3$ & & $w_1$ & $w_3$\\
      $m_2$ & $m_5$ & & $w_2$ & $w_5$\\
      $m_3$ & $m_1$ & & $w_3$ & $w_1$\\
      $m_4$ & $m_4$ & & $w_4$ & $w_4$\\
      $m_5$ & $m_2$ & & $w_5$ & $w_2$\\
    \end{tabular}
    \quad
    \begin{tabular}{c|c c c|c}
      Man & List & & Woman & List\\
      \hline
      $m_1$ & $\pi_1$ & & $w_1$ & $\sigma_2$\\
      $m_2$ & $\pi_1$ & & $w_2$ & $\sigma_2$\\
      $m_3$ & $\pi_2$ & & $w_3$ & $\sigma_1$\\
      $m_4$ & $\pi_1$ & & $w_4$ & $\sigma_2$\\
      $m_5$ & $\pi_2$ & & $w_5$ & $\sigma_1$\\
    \end{tabular}
  \end{center}
\end{table}

As a remark, this ``greedy'' approach where we select a man, find his
most preferred available woman, and permanently match him to her will
not work in general. Table \ref{tab:preferences} describes a simple
2-list example where the unique stable matching is
$\{(m_1,w_2),(m_2,w_3),(m_3,w_5),(m_4,w_4),(m_5,w_1)\}$. In this instance, no
participant is matched with their top choice. Therefore, the above
approach cannot work for this instance. This illustrates to some
extent why the general case seems more difficult than the one-sided
case.

An alternative model of a greedy approach that is based on work by
Davis and Impagliazzo in \cite{DavisImpagliazzo2009} also will not
work. In this model, an algorithm can view each of the lists and the
preferences of the women. It can then (adaptively) choose an order in
which to process the men. When processing a man, he must be assigned a
partner (not necessarily his favorite available woman) once and for
all, based only on his choice of preference list and the preferences
of the previously processed men. This model is similar to online
stable matching \cite{KhullerMitchellVazirani1994} except that it
allows the algorithm to choose the processing order of the men. Using
the preferences in Table \ref{tab:greedylist} and minor modifications
to them, we can show that no greedy algorithm of this type can
successfully produce a stable matching. Indeed, the unique stable
matching of the preference scheme below is
$\mu = \{(m_1,w_3),(m_2,w_1),(m_3,w_2)\}$. However, changing the
preference list for whichever of $m_1$ or $m_2$ is processed later
will form a blocking pair with the stable partner of the other. If
$m_1$ uses $\pi_1$, $(m_1,w_1)$ blocks $\mu$ and if $m_2$ uses
$\pi_2$, $(m_2,w_3)$ blocks $\mu$. Therefore, no algorithm can succeed
in assigning stable partners to these men without first knowing the
preference list choice of all three.

\begin{table}[H]
  \caption{Two-list preferences where a greedy approach will not work}
  \label{tab:greedylist}
  \begin{center}
    \begin{tabular}{c|c c c|c}
      $\sigma_1$ & $\sigma_2$ & & $\pi_1$ & $\pi_2$\\
      \hline
      $m_1$ & $m_2$ & & $w_1$ & $w_3$\\
      $m_2$ & $m_1$ & & $w_2$ & $w_2$\\
      $m_3$ & $m_3$ & & $w_3$ & $w_1$\\
    \end{tabular}
    \quad
    \begin{tabular}{c|c c c|c}
      Man & List & & Woman & List\\
      \hline
      $m_1$ & $\pi_2$ & & $w_1$ & $\sigma_1$\\
      $m_2$ & $\pi_1$ & & $w_2$ & $\sigma_2$\\
      $m_3$ & $\pi_1$ & & $w_3$ & $\sigma_2$\\
    \end{tabular}
  \end{center}
\end{table}

\subsection{Strategic Behavior}
As mentioned earlier, strategic behavior in the general preference setting allows for participants to truncate or rearrange their lists. However, in the $d$-attribute and $d$-list models, we assume that the attributes or lists are fixed, so that the only manipulation the participants are allowed is to misrepresent their weight vectors or which list they choose. Despite this limitation, there is still no strategy proof mechanism for finding a stable matching when $d\geq 2$.

\begin{theorem}\label{thm:liststrategy}
  For $d \geq 2$ there is no strategy proof algorithm to find a stable matching in the $d$-list model.
\end{theorem}

\begin{proof}
Table \ref{tab:liststrategy} describes true preferences that can be manipulated by the women. Observe that there are two stable matchings: the man-optimal matching $\{(m_1,w_1),(m_2,w_2),(m_3,w_3),(m_4,w_4)\}$ and the woman-optimal matching $\{(m_1,w_2),(m_2,w_3),(m_3,w_1),(m_4,w_4)\}$. However, if $w_2$ used list $\sigma_2$ instead of $\sigma_1$, then there is a unique stable matching which is $\{(m_1,w_2),(m_2,w_3),(m_3,w_1),(m_4,w_4)\}$, the woman-optimal stable matching from the original preferences. Therefore, any mechanism that does not always output the woman optimal stable matching can be manipulated by the women to their advantage. By symmetry, any mechanism that does not always output the man-optimal matching could be manipulated by the men. Thus there is no strategy-proof mechanism for the $d$-list setting with $d\geq 2$.
\end{proof}

\begin{table}[H]
  \caption{Two-list preferences that can be manipulated}
  \label{tab:liststrategy}
  \begin{center}
    \begin{tabular}{c|c c c|c}
      $\sigma_1$ & $\sigma_2$ & & $\pi_1$ & $\pi_2$\\
      \hline
      $m_1$ & $m_3$ & & $w_1$ & $w_3$\\
      $m_2$ & $m_1$ & & $w_2$ & $w_1$\\
      $m_3$ & $m_4$ & & $w_3$ & $w_2$\\
      $m_4$ & $m_2$ & & $w_4$ & $w_4$\\
    \end{tabular}
    \quad
    \begin{tabular}{c|c c c|c}
      Man & List & & Woman & List\\
      \hline
      $m_1$ & $\pi_1$ & & $w_1$ & $\sigma_2$\\
      $m_2$ & $\pi_1$ & & $w_2$ & $\sigma_1$\\
      $m_3$ & $\pi_2$ & & $w_3$ & $\sigma_1$\\
      $m_4$ & $\pi_2$ & & $w_4$ & $\sigma_1$\\
    \end{tabular}
  \end{center}
\end{table}

Since the $d$-list model is a special case of the $d$-attribute model, we immediately have the following result from Theorem \ref{thm:liststrategy}.

\begin{corollary}
  For $d \geq 2$ there is no strategy proof algorithm to find a stable matching in the $d$-attribute model.
\end{corollary}

Of course in the 1-list setting there is a trivial unique stable matching. Moreover, in the one-sided $d$-attribute model our algorithm is strategy proof since the women are receiving the woman-optimal matching and each man receives his best available woman, so misrepresentation would only give him a worse partner.

\section{Verification}

We now turn to the problem of verifying whether a given matching is
stable. While this is as hard as finding a stable matching in the
general setting, the verification algorithms we present here are more
efficient than our algorithms for finding stable matchings in the
attribute model.

\subsection{Real Attributes and Weights}
\label{sec:verifyreal}

In this section we adapt the geometric approach for finding a stable
matching in the one-sided $d$-attribute model to the problem of
verifying a stable matching in the (two-sided) $d$-attribute model. We
express the verification problem as a \emph{simplex range searching
  problem} in $\reals^{2d}$, which is the dual of the ray shooting
problem. In simplex range searching we are given $n$ points and answer
queries that ask for the number of points inside a simplex. In our
case we only need degenerate simplices consisting of the intersection
of two halfspaces. Simplex range searching queries can be done in
sublinear time for constant $d$.

\begin{lemma}[\cite{Matousek1992}]
  \label{lem:range}
  Given a set of $n$ points in $\reals^d$, one can process it for simplex
  range searching in time $O(n \log n)$, and then answer queries in
  time $\tilde{O}(n^{1-\frac1d})$.
\end{lemma}

For $1 \leq d' \leq d$ we use the notation
$(x_1, \ldots, x_{d}, y_1, \ldots, y_{d'-1}, z)$ for points in
$\reals^{d + d'}$. We again let $\dim(w)$ be the index of $w$'s last
non-zero weight, assume without loss of generality
$\alpha_{\dim(w)} \in \{-1, 1\}$, and let
$\sgn(w) = \sgn(\alpha_{\dim(w)})$. We partition the set of women into
$2d$ sets $W_{d',s}$ for $1 \leq d' \leq d$ and $s \in \{-1, 1\}$
based on $\dim(w)$ and $\sgn(w)$. Note that if $\dim(w) = 0$, then $w$
is indifferent among all men and can therefore not be part of a blocking
pair. We can ignore such women.

For a woman $w$, consider the point
\begin{equation}
  P(w) = (A_1(w), \ldots, A_{d}(w), \alpha_1(w), \ldots,
  \alpha_{\dim(w)-1}(w), \val_{w}(m))
\end{equation}
where $m=\mu(w)$ is the partner of $w$ in the input matching $\mu$. For a set $W_{d',s}$ we let $P_{d',s}$ be the set of points $P(w)$ for
$w \in W_{d',s}$. The basic idea is to construct a simplex for every
man and query it against all sets $P_{d',s}$. 

Given $d'$,$s$, and a man $m$, let $H_1(m)$ be the halfspace
$\left\{\sum_{i=1}^{d} \alpha_i(m) x_i > \val_m(w)\right\}$ where $w=\mu(m)$. For $w' \in W_{d',s}$ we have
$P(w') \in H_1(m)$ if and only if $m$ strictly prefers $w'$ to $w$. Further let $H_2(m)$ be the halfspace $\left\{\sum_{i=1}^{d'-1} A_i(m)
  y_i + A_{d'}(m) s > z\right\}$. For $w' \in W_{d',s }$ we have $P(w')
\in H_2(m)$ if and only if $w'$ strictly prefers $m$ to $\mu(w')$. Hence $(m,w')$ is a blocking pair if and only if $P(w') \in H_1(m) \cap
H_2(m)$.

Using Lemma \ref{lem:range} we immediately have an algorithm to verify a stable matching.

\begin{theorem}
  \label{thm:verifyreal}
  There is an algorithm to verify a stable matching in the
  $d$-attribute model with real-valued attributes and weights in time
  $\tilde{O}(n^{2-1/2d})$
\end{theorem}
\begin{proof}
  Partition the set of women into sets $W_{d',s}$ for $1 \leq d' \leq
  d$ and $s \in \{-1,1\}$ and for $w \in W_{d',s}$ construct $P(w) \in \reals^{d + d'}$ as above. Then preprocess the sets according to Lemma \ref{lem:range}.
  For each man $m$ query $H_1(m) \cap H_2(m)$ against the points
  in all sets. By the definitions of $H_1(m)$ and $H_2(m)$, there is a blocking pair if and only if for some man $m$ there is a point $P(w) \in H_1(m) \cap H_2(m)$ in one of the sets $P_{d',s}$.

  The time to preprocess is $O(n \log n)$. There are $2dn$ queries of
  time $\tilde{O}(n^{1-1/2d})$. Hence the whole process
  requires time $\tilde{O}(n^{2-1/2d})$ as claimed.
\end{proof}

\begin{algorithm}[H]
// For points in $P \in \reals^{d+d'}$ we use the notation $(x_1,
\ldots, x_{d}, y_1, \ldots, y_{d'-1},z)$ to refer to its coordinates. \\

Input: matching $\mu$\\
\For{each woman $w$}{
  $m \gets \mu(w)$\\
  $P(w) \gets (A_1(w), \ldots, A_{d}(w), \alpha_1(w), \ldots, \alpha_{d}(w),
  \val_{w}(m))$ \\
  $P_{\dim(w), \sgn(w)} \gets W_{\dim(w), \sgn(w)} \cup P(w)$ 
}

\For{$d' =1$ to $d$ and $s \in \{-1,1\}$} {
  $P_{d',s}$.preprocess()\\
  \For{each man $m$}{
    $w \gets \mu(m)$\\
    $H_1(m) \gets \left\{\sum_{i=1}^{d} \alpha_i(m) x_i > \val_m(w)\right\}$\\
    $H_2(m) \gets \left\{\sum_{i=1}^{d'-1} A_i(m) y_i + A_{d'}(m)
    \cdot s > z\right\}$\\
    \If{$\Query(P_{d',s}, H_1(m) \cap H_2(m)) > 0$}{
      \Return{$\mu$ is not stable}
    }
  }
}
\Return{$\mu$ is stable}
\caption{Verify Stable Matching with Reals}
\label{alg:testreals}
\end{algorithm}

\subsection{Lists}
\label{sec:lists}

When there are $d$ preference orders for each side, and each
participant uses one of the $d$ lists, we provide a more efficient
algorithm. Here, assume $\mu$ is the given matching between $M$ and
$W$. Let $\{\pi_i\}_{i=1}^d$ be the set of $d$ permutations on the
women and $\{\sigma_i\}_{i=1}^d$ be the set of $d$ permutations on the
men. Define $\rank(w,i)$ to be the position of $w$ in permutation
$\pi_i$. This can be determined in constant time after $O(dn)$
preprocessing of the permutations. Let $\head(\pi_i,j)$ be the first
woman in $\pi_i$ who uses permutation $\sigma_j$ and $\nxt(w,i)$ be
the next highest ranked woman after $w$ in permutation $\pi_i$ who
uses the same permutation as $w$ or $\bot$ if no such woman
exists. These can also be determined in constant time after $O(dn)$
preprocessing by splitting the lists into sublists, with one sublist
for the women using each permutation of men. The functions $\rank$, $\head$, and
$\nxt$ are defined analogously for the men.

\begin{algorithm}[H]
\For{$i = 1$ to $d$}{
	\For{$j = 1$ to $d$}{
		$w \gets \head(\pi_i,j)$.\\
		$m \gets \head(\sigma_j,i)$.\\
		\While{$m\neq \bot$ and $w\neq \bot$}{
			\If{$\rank(w,i) > \rank(\mu(m),i)$}{
				$m \gets \nxt(m,j)$.
			}
			\Else{
				\If{$\rank(m,j) > \rank(\mu(w),j)$}{
					$w \gets \nxt(w,i)$.
				}
				\Else{
					\Return{$(m,w)$ is a blocking pair.}
				}
			}
		}
	}
}
\Return{$\mu$ is stable.}
\caption{Verify $d$-List Stable Matching}
\label{alg:listtestalgorithm}
\end{algorithm}

\begin{theorem}
  There is an algorithm to verify a stable matching in the $d$-list
  model in $O(dn)$ time.
\end{theorem}
\begin{proof}
  We claim that algorithm \ref{alg:listtestalgorithm} satisfies the
  theorem. Indeed, if the algorithm returns a pair $(m,w)$ where $m$
  uses $\pi_i$ and $w$ uses $\sigma_j$, then $(m,w)$ is a blocking
  pair because $w$ appears earlier in $\pi_i$ than $\mu(m)$ and $m$
  appears earlier in $\sigma_j$ than $\mu(w)$.

On the other hand, suppose the algorithm returns that $\mu$ is stable
but there is a blocking pair, $(m,w)$, where $m$ uses $\pi_i$ and $w$
uses $\sigma_j$. The algorithm considers permutations $\pi_i$ and
$\sigma_j$ since it does not terminate early. Clearly if the algorithm
evaluates $m$ and $w$ simultaneously when considering permutations
$\pi_i$ and $\sigma_j$, it will detect that $(m,w)$ is a blocking
pair. Therefore, the algorithm either moves from $m$ to $\nxt(m,j)$
before considering $w$ or it moves from $w$ to $\nxt(w,i)$ before
considering $m$. In the former case, $\rank(\mu(m),i) < \rank(w',i)$
for some $w'$ that comes before $w$ in $\pi_i$. Therefore $m$ prefers
$\mu(m)$ to $w$. Similarly, in the latter case, $\rank(\mu(w),j) <
\rank(m',i)$ for some $m'$ that comes before $m$ in $\sigma_j$ so $w$
prefers $\mu(w)$ to $m$. Thus $(m,w)$ is not a blocking pair and we
have a contradiction.

The \textbf{for} and \textbf{while} loops proceed through all men and
women once for each of the $d$ lists in which they appear. Since 
at each step we are either proceeding to the next man
or the next woman unless we find a blocking pair, the algorithm requires time 
$O(dn)$. This is optimal since the input size is $dn$.
\end{proof}

\subsection{Boolean Attributes and Weights}
\label{sec:verifyboolean}

In this section we consider the problem of verifying a stable matching
when the $d$ attributes and weights are restricted to boolean values
and $d = c \log n$. The algorithm closely follows an algorithm for the
maximum inner product problem by Alman and Williams
\cite{AlmanWilliams2015}. The idea is to express the existence of a
blocking pair as a probabilistic polynomial with a bounded number of
monomials and use fast rectangular matrix multiplication to evaluate
it. A probabilistic polynomial for a function $f$ is a polynomial $p$
such that for every input $x$
\begin{equation}
  \label{eq:3}
  \Pr[f(x) \neq p(x)] \leq \frac13
\end{equation}

We use the following tools in our algorithm. $\maj_d$ is the threshold function that outputs $1$ if at least $d$ of its inputs are $1$.

\begin{lemma}[\cite{AlmanWilliams2015}]
\label{lem:maj}
  There is a probabilistic polynomial for $\maj_d$ on $n$ variables
  and error $\varepsilon$ with degree $O(\sqrt{n \log(1/\varepsilon)})$.
\end{lemma}

\begin{lemma}[\cite{Razborov1987, Smolensky1987}]
  \label{lem:razsmol}
  There is a probabilistic polynomial for the disjunction of $n$
  variables and error $\varepsilon$ with degree $O(\log(1/\varepsilon))$ 
\end{lemma}
 
\begin{lemma}[\cite{Williams2014}]
  \label{lem:eval}
  Given a polynomial $P(x_1, \ldots, x_m, y_1, \ldots y_m)$ with at
  most $n^{0.17}$ monomials and two
  sets $X, Y \subseteq \{0,1\}^m$ with $|X| = |Y| = n$, we can
  evaluate $P$ on all pairs $(x,y) \in X \times Y$ in time
  $\tilde{O}(n^2 + m \cdot n^{1.17})$. 
\end{lemma}

We construct a probabilistic polynomial that outputs
$1$ if there is a blocking pair. To minimize the degree of the polynomial, we pick a parameter $s$ and divide the men and women into sets
of size at most $s$. The polynomial takes the description of $s$ men $m_{1}, \ldots, m_{s}$ and $s$ women
$w_{1}, \ldots, w_{s}$ along with their respective partners as input, and
outputs $1$ if and only if there is a blocking pair $(m_i,w_j)$ among
the $s^2$ pairs of nodes with high probability.

\begin{lemma}
  \label{lem:probpol}
  Let $u$ be a large constant and $s=n^{1/u c \log^2 c}$. There
  is a probabilistic polynomial with the following inputs:
  \begin{itemize}
  \item The attributes and weights of $s$ men, $A(m_{1}), \ldots,
    A(m_{s}), \alpha(m_{1}), \ldots, \alpha(m_{s})$
  \item The attributes of the $s$ women that are matched with these men
    $A(\mu(m_1)), \ldots, A(\mu(m_s))$
  \item The attributes and weights of $s$ women, $A(w_{1}), \ldots,
    A(w_{s}), \alpha(w_{1}), \ldots, \alpha(w_{s})$
  \item The attributes of the $s$ men that are matched with these women
    $A(\mu(w_1)), \ldots, A(\mu(w_s))$
  \end{itemize}
  The output of the polynomial is $1$ if and only if there is a
  blocking pair with respect to the matching $\mu$ among the $s^2$
  pairs in the input. The number of monomials is at most $n^{0.17}$
  and the polynomial can be constructed efficiently.
\end{lemma}
\begin{proof}
  A pair $(m_i, w_j)$ is a blocking pair if and only if
  $\langle\alpha(m_i), A(\mu(m_i))\rangle < \langle\alpha(m_i), A(w_j)\rangle$ and 
  $\langle\alpha(w_j), A(\mu(w_j))\rangle < \langle\alpha(w_j),
  A(m_i)\rangle$. 
  Rewriting
  \begin{equation}
    \label{eq:2}
    F(x,y,a,b) := \langle x, y\rangle < \langle a, b \rangle =
   \maj_{d+1}\left(\neg(x_1 \land y_1), \ldots, \neg(x_d \land y_d), a_1 \land
     b_1, \ldots, a_d \land b_d\right)
  \end{equation}

  we have a blocking pair if and only if 
  \begin{equation}
    \label{eq:1}
    \bigvee_{\substack{i \in [1, s] \\ j \in [1, s]}}
    \big(F(\alpha(m_i), A(\mu(m_i)), \alpha(m_i), A(w_j)) \land F(\alpha(w_j), A(\mu(w_j)), \alpha(w_j), A(m_i))\big)
  \end{equation}

  Note that we can easily adapt this algorithm to
  finding strongly blocking pairs by defining $F(x,y,a,b)$ as $\langle
  x, y\rangle \leq \langle a, b \rangle$.

  Using Lemma \ref{lem:maj} with $\varepsilon = \frac{1}{s^3}$ and
  Lemma \ref{lem:razsmol} with $\varepsilon = 1/4$ we get a
  probabilistic polynomial of degree $a \sqrt{d \log s}$ for some
  constant $a$ and error $1/4 + 1/s < 1/3$. Furthermore, since we are
  only interested in boolean inputs we can assume the polynomial to be
  multilinear. For large enough $u$ we have $2d > a \sqrt{d \log(s)}$
  (i.e.~the degree is at most half of the number of variables)
  and the number of monomials is then bounded by
  $O\left(\left(s^2\binom{4d}{a\sqrt{d \log(s)}}\right)^2\right)$.

  Simplifying the binomial coefficient we have
  \begin{equation*}
    \binom{4d}{a \sqrt{d \log s}} = \binom{4 c \log n}{a \sqrt{(\log^2
        n)/u \log^2 c}} = \binom{4 c \log n}{a \log n / \sqrt{u}\log
      c}
  \end{equation*}

  Setting $\delta = a/(\sqrt{u}\log(c))$ we can upper bound this using
  Stirling's inequality by
  \begin{equation*}
    \binom{4 c \log n}{\delta \log n} \leq \left(\frac{(4 c \log n)\cdot
        e}{\delta}\right)^{\delta \log n} = n^{\delta \log(4ce/\delta)}
  \end{equation*}

  By choosing $u$ to be a large enough constant, we can make $\delta$ and the
  exponent arbitrarily small. The factor of $s^2$ only
  contributes a trivial constant to the exponent. Therefore we can bound
  the number of monomials by $n^{0.17}$.
\end{proof}

\begin{theorem}
  \label{thm:verifyboolean}
  In the $d$-attribute model with $n$ men and women, and $d = c \log
  n$ boolean attributes and weights, there is a randomized algorithm to decide
  if a given matching is stable in time $\tilde{O}(n^{2-1/O(c \log^2(c))})$ with 
  error probability at most $1/3$.
\end{theorem}
\begin{proof}
  We again choose $s=n^{1/u c \log^2 c}$ and construct the probabilistic
  polynomial as in Lemma \ref{lem:probpol}. We then divide the men and
  women into $\lceil \frac{n}{s} \rceil$ groups of size at most $s$.

  For a group of men $m_{1}, \ldots, m_{s}$ we let the
  corresponding input vector be 
  \begin{equation*}
  A(m_{1}), \ldots, A(m_{s}),
  \alpha(m_{1}), \ldots, \alpha(m_{s}), A(\mu(m_1)), \ldots,
  A(\mu(m_s))
  \end{equation*}
  We set $X$ as the set of all input vectors for the
  $\lceil \frac{n}{s} \rceil$ groups. We define the set $Y$
  symmetrically for the input vectors corresponding to the $\lceil
  \frac{n}{s} \rceil$ groups of women.

  Using Lemma \ref{lem:eval} we evaluate the polynomial on all pairs
  $x \in X$, $y\in Y$ in time
  \begin{equation}
    \tilde{O}\left(\left(\frac{n}{s}\right)^2 +
    O(sd)\left(\frac{n}{s}\right)^{1.17}\right) =
  \tilde{O}\left(\left(\frac{n}{s}\right)^2\right) =
  \tilde{O}(n^{2-1/O(c \log^2(c))})
  \end{equation}
  The probability that the output is wrong for any fixed input pair is
  at most $1/3$. We repeat this process $O(\log n)$ times and take the
  threshold output for every pair of inputs, such that the error
  probability is at most $O\left(\frac{1}{n^2}\right)$ for any fixed
  pair of inputs. Using a union bound we can make the probability of
  error at most $1/3$ on any input.
\end{proof}

\section{Conditional Hardness}
\label{sec:hardness}
\subsection{Background}
The Strong Exponential Time Hypothesis has proved useful in
arguing conditional hardness for a large number of problems. We show $\SETH$-hardness for both verifying and finding a
stable matching in the $d$-attribute model, even if the weights and
attributes are boolean. The main step of the proof is a reduction from
the maximum inner product problem to the stable matching problem. The maximum inner product problem is known to be $\SETH$-hard. We give
the fine-grained reduction from $\cnfsat$ to the vector orthogonality
problem and from the vector orthogonality
problem to the maximum inner product problem for the sake of completeness.

\begin{definition}[\cite{ImpagliazzoPaturi2001, ImpagliazzoPaturiZane2001}]
  The \emph{Strong Exponential Time Hypothesis} ($\SETH$) stipulates that
  for each $\varepsilon > 0$ there is a $k$ such that $\ksat$ requires
  time $\Omega(2^{(1-\varepsilon) n})$.
\end{definition}

\begin{definition}
  For any $d$, the vector orthogonality problem
  is to decide if two input sets $U, V
  \subseteq \reals^d$ with $|U| = |V| = n$ have a pair $u \in U$, $v
  \in V$ such that $\langle u, v \rangle = 0$.

  The boolean vector orthogonality problem is the variant where $U, V
  \subseteq \{0,1\}^d$.
\end{definition}

\begin{definition}
  For any $d$ and input $l$, the maximum inner product problem
  is to decide if two input sets $U, V
  \subseteq \reals^d$ with $|U| = |V| = n$ have a pair $u \in U$, $v
  \in V$ such that $\langle u, v \rangle \geq l$.

  The boolean maximum inner product problem is the variant where $U, V
  \subseteq \{0,1\}^d$.
\end{definition}

\begin{lemma}[\cite{ImpagliazzoPaturiZane2001,Williams2004,AlmanWilliams2015}]
  Assuming $\SETH$, for any $\varepsilon > 0$, there is a $c$ such
  that solving the boolean maximum inner product problem on $d = c
  \log n$ dimensions requires time $\Omega(n^{2-\varepsilon})$.
\end{lemma}
\begin{proof}
  The proof is a series of reductions from $\ksat$ to boolean inner
  product. By the Sparsification Lemma \cite{ImpagliazzoPaturiZane2001} we can
  reduce $\ksat$ to a subexponential number of $\ksat$ instances with
  at most $d = c_k n$ clauses, where $c_k$ does not depend on
  $n$. Hence, assuming $\SETH$, for any $\varepsilon > 0$, there is a $c$
  such that $\cnfsat$ with $cn$ clauses requires time
  $\Omega(2^{(1-\varepsilon)n})$.

  We reduce $\cnfsat$ to the boolean vector orthogonality problem
  using a technique called \emph{Split and List}. Divide the variable
  set into two sets $S, T$ of size $\frac{n}{2}$ and for each
  set consider all $N = 2^{n/2}$ assignments to the
  variables. For every assignment we construct a $d$-dimensional vector
  where the $i$th position is $1$ if and only if the assignment does
  not satisfy the $i$th clause of the CNF formula. Let $U$ be the
  set of vectors corresponding to the assignments to $S$ and let $V$ be
  the set of vectors corresponding to $T$. A pair $u \in U$, $v \in V$
  is orthogonal if and only if the corresponding assignment satisfies
  all clauses. An algorithm for boolean vector orthogonality in
  dimension $d = c n = 2 c \log N$ and time $O(N^{2-\varepsilon}) =
  O(2^{(1-\varepsilon/2)n})$ would contradict $\SETH$. Hence
  assuming $\SETH$, for every $\varepsilon > 0$ there is a $c$ such
  that the boolean vector orthogonality problem with $d = c \log N$
  requires time $\Omega(N^{2-\varepsilon})$.

  Finally, we reduce the boolean vector orthogonality problem
  to the boolean maximum inner product problem by partitioning the set
  $U$ into sets $U_i$ for $0 \leq i \leq d$ where $U_i$ contains all
  vectors with Hamming weight $i$. Observe that a vector $v \in V$ is
  orthogonal to a vector $u \in U_i$ if and only if $\langle u,
  \neg v \rangle = i$, where $\neg v$ is the element-wise complement
  of $v$. Thus $U$ and $V$ have an orthogonal pair,
  if and only if there is an $i$ such that $U_i$ and $\neg V
  =\{\neg v \; |\; v \in V\}$ have a pair with inner product at
  least $i$. Therefore, for any $\varepsilon > 0$ there is a $c$ such
  that the maximum inner product problem on $d = c \log N$ dimensions
  requires time $\Omega(N^{2-\varepsilon})$ assuming $\SETH$.  
\end{proof}

\subsection{Finding Stable Matchings}
In this subsection we give a fine-grained reduction from the maximum
inner product problem to the problem of finding a stable matching in
the boolean $d$-attribute model. This shows that the stable matching
problem in the $d$-attribute model is $\SETH$-hard, even if we
restrict the attributes and weights to booleans.

\begin{theorem}
  \label{thm:hardnessfinding}
  Assuming $\SETH$, for any $\varepsilon > 0$, there is a $c$ such that
  finding a stable matching in the boolean $d$-attribute model
  with $d = c \log n$ dimensions requires time
  $\Omega(n^{2-\varepsilon})$.
\end{theorem}
\begin{proof}
  The proof is a reduction from maximum inner product to finding a
  stable matching. Given an instance of the maximum inner
  product problem with sets $U, V \subseteq \{0,1\}^d$ where
  $|U| = |V| = n$ and threshold $l$, we construct a matching market
  with $n$ men and $n$ women. For every $u \in U$ we have a man $m_u$
  with $A(m_u) = u$ and $\alpha(m_u) = u$. Similarly, for vectors
  $v \in V$ we have women $w_v$ with $A(w_v) = v$ and
  $\alpha(w_v) = v$. This matching market is symmetric in the sense
  that for $m_u$ and $w_v$,
  $\val_{m_u}(w_v) = \val_{w_v}(m_u) = \langle u, v \rangle$.

  We claim that any stable matching contains a pair $(m_u,w_v)$
  such that the inner product $\langle u, v \rangle$ is
  maximized. Indeed, suppose there are vectors $u \in U$, $v \in V$ with
  $\langle u, v \rangle \geq l$ but there exists a stable matching $\mu$
  with $\langle
  u', v' \rangle < l$ for all pairs $(m_{u'}, w_{v'}) \in \mu$. Then $(m_u, w_v)$ is clearly a blocking pair for $\mu$ which is a contradiction.
\end{proof}

\subsection{Verifying Stable Matchings}
In this section we give a reduction from the maximum inner product
problem to the problem of verifying a stable matching, showing that
this problem is also $\SETH$-hard.

\begin{theorem}
  \label{thm:hardnessverify}
  Assuming $\SETH$, for any $\varepsilon > 0$, there is a $c$ such that
  verifying a stable matching in the boolean $d$-attribute
  model with $d = c \log n$ dimensions requires time
  $\Omega(n^{2-\varepsilon})$.
\end{theorem}
\begin{proof}
  We give a reduction from maximum inner product with sets
  $U, V \subseteq \{0,1\}^d$ where $|U| = |V| = n$ and threshold
  $l$. We construct a matching market with $2n$ men and women in the
  $d'$-attribute model with $d' = d + 2(l-1)$. Since $d' < 3d$ the
  theorem then follows immediately from the $\SETH$-hardness of
  maximum inner product.

  For $u \in U$, let $m_u$ be a man in the matching market with
  attributes and weights $A(m_u) = \alpha(m_u) = u \circ 1^{l-1} \circ
  0^{l-1}$ where we use $\circ$ for concatenation. Similarly, for $v
  \in V$ we have a woman $w_v$ with $A(w_v) = \alpha(w_v) = v \circ
  0^{l-1} \circ 1^{l-1}$. We further introduce \emph{dummy women}
  $w'_u$ for $u \in U$ with $A(w'_u) = \alpha(w'_u) = 0^d \circ
  1^{l-1} \circ 0^{l-1}$ and \emph{dummy men} $m'_v$ for $v \in V$
  with $A(m'_v) = \alpha(m'_v) = 0^d \circ 0^{l-1} \circ 1^{l-1}$.

  We claim that the matching consisting of pairs $(m_u, w'_u)$ for all
  $u \in U$ and $(m'_v, w_v)$ for all $v \in V$ is stable if
  and only if there is no pair $u \in U$, $v \in V$ with
  $\langle u, v \rangle \geq l$. For $u, u' \in U$ we have
  $\val_{m_u}(w'_{u'}) = \val_{w'_{u'}}(m_u) = l - 1$, and for
  $v, v' \in V$ we have
  $\val_{w_v}(m'_{v'}) = \val_{m'_{v'}}(w_v) = l -1$. In particular,
  any pair in $\mu$ has (symmetric) value $l-1$. Hence there is a
  blocking pair with respect to $\mu$ if and only if there is a pair
  with value at least $l$. For $u \neq u'$ and $v \neq v'$ the pairs
  $(m_u, w'_{u'})$ and $(w_v, m'_{v'})$ can never be blocking pairs as
  their value is $l-1$. Furthermore for any pair of dummy nodes $w'_u$
  and $m'_v$ we have $\val_{m'_v}(w'_u) = \val_{w'_u}(m'_v) = 0$, thus
  no such pair can be a blocking pair either. This leaves pairs of
  real nodes as the only candidates for blocking pairs. For non-dummy
  nodes $m_u$ and $w_v$ we have
  $\val_{m_u}(w_v) = \val_{w_v}(m_u) = \langle u, v \rangle$ so
  $(m_u,w_v)$ is a blocking pair if and only if
  $\langle u, v \rangle \geq l$.
\end{proof}

\tikzstyle{agent} = [draw, thin, circle, minimum size=1.15cm]
\tikzstyle{man} = [agent, fill=red!20]
\tikzstyle{woman} = [agent, fill=blue!20]
\tikzstyle{list} = []
\tikzstyle{match} = [draw, thick]
\tikzstyle{blocking} = [match, red]

\begin{figure}[H]
\caption{A representation of the reduction from maximum inner product
  to verifying a stable matching}
\centering
\begin{tikzpicture}[node distance=0.5cm, auto,>=latex']
    \path node[man, xshift = 2cm] (m1) {$m_{u_1}$}
              node[man, below = 0.1cm of m1] (m2) {$m_{u_2}$}
              node[man, below = 0.1cm of m2] (m3) {$m_{u_3}$};
    \path node[woman, right = of m1, xshift = 4cm] (dw1) {$w'_{u_1}$}
              node[woman, below = 0.1cm of dw1] (dw2) {$w'_{u_2}$}
              node[woman, below = 0.1cm of dw2] (dw3) {$w'_{u_3}$};

    \path node[woman, below = 0.1cm of dw3] (w1) {$w_{v_1}$}
              node[woman, below = 0.1cm of w1] (w2) {$w_{v_2}$}
              node[woman, below = 0.1cm of w2] (w3) {$w_{v_3}$};
   \path node[list, left = of m1, xshift = 0.5cm] (pm1n)
                  {$u_1 \circ 1^{l-1} \circ 0^{l-1}$}
                  node[list, left = of m2, xshift = 0.5cm] (pm2n)
                  {$u_2 \circ 1^{l-1} \circ 0^{l-1}$}
                  node[list, left = of m3, xshift = 0.5cm] (pm3n)
                  {$u_3 \circ 1^{l-1} \circ 0^{l-1}$}
                  node[list, right = of w1, xshift = -0.5cm] (pw1n)
                  {$v_1 \circ 0^{l-1} \circ 1^{l-1}$}
                  node[list, right = of w2, xshift = -0.5cm] (pw2n)
                  {$v_2 \circ 0^{l-1} \circ 1^{l-1}$}
                  node[list, right = of w3, xshift = -0.5cm] (pw3n)
                  {$v_3 \circ 0^{l-1} \circ 1^{l-1}$};  
    \path     node[list, right = of dw1, xshift = -0.5cm] (pdw1)
                  {$0^d \circ 1^{l-1} \circ 0^{l-1}$}
                  node[list, right = of dw2, xshift = -0.5cm] (pdw2)
                  {$0^d \circ 1^{l-1} \circ 0^{l-1}$}
                  node[list, right = of dw3, xshift = -0.5cm] (pdw3)
                  {$0^d \circ 1^{l-1} \circ 0^{l-1}$};
    \path node[man, below = 0.1cm of m3] (dm1) {$m'_{v_1}$}
              node[man, below = 0.1cm of dm1] (dm2) {$m'_{v_2}$}
              node[man, below = 0.1cm of dm2] (dm3) {$m'_{v_3}$}
              node[list, left = of dm1, xshift = 0.5cm] (pdm1) {$0^{d}
                \circ 0^{l-1} \circ 1^{l-1}$}
              node[list, left = of dm2, xshift = 0.5cm] (pdm2) {$0^{d} \circ 0^{l-1} \circ 1^{l-1}$}
              node[list, left = of dm3, xshift = 0.5cm] (pdm3)
              {$0^{d} \circ 0^{l-1} \circ 1^{l-1}$};
     \path (m1) edge[match] (dw1)
               (m2) edge[match] (dw2)
               (m3) edge[match] (dw3)
               (dm1) edge[match] (w1)
               (dm2) edge[match] (w2)
               (dm3) edge[match] (w3);
\end{tikzpicture}
\end{figure}
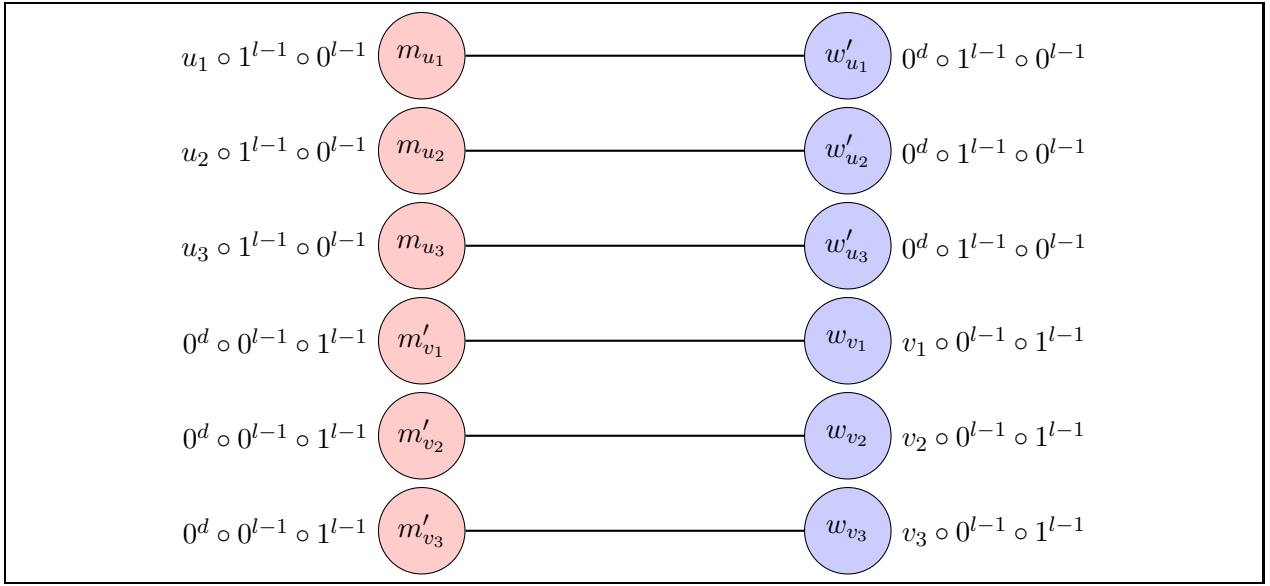

\subsection{Checking a Stable Pair}

In this section we give a reduction from the maximum inner product
problem to the problem of checking whether a given pair is part of any
or all stable matchings, showing that these questions are $\SETH$-hard
when $d = c \log n$ for some constant $c$. For general preferences,
both questions can be solved in time $O(n^2)$ \cite{IrvingLeather1986,Gusfield1987}  and are known
to require quadratic time \cite{ng1990lower,segal2007communication,gonczarowski2015stable}.

\begin{theorem}\label{thm:stablepair}
  Assuming $\SETH$, for any $\varepsilon > 0$, there is a $c$ such that
  determining whether a given pair is part of any or all stable matchings in the boolean $d$-attribute
  model with $d = c \log n$ dimensions requires time
  $\Omega(n^{2-\varepsilon})$.
\end{theorem}
\begin{proof}
  We again give a reduction from maximum inner product with sets
  $U, V \subseteq \{0,1\}^d$ where $|U| = |V| = n$ and threshold
  $l$. We construct a matching market with $2n$ men and women in the
  $d'$-attribute model with $d' = 7d + 7(l-1)+ 18$. Since $d' < 15d$ the
  theorem then follows immediately from the $\SETH$-hardness of
  maximum inner product.

  For simplicity, we will first describe the preference scheme, then provide weight and attribute vectors that result in those preferences. For $u \in U$, let $m_u$ be a man in the matching market and for $v
  \in V$ we have a woman $w_v$. We also have $n-1$ \emph{dummy men} $m_i:i=1\dots n-1$ and $n-1$ \emph{dummy women} $w_j:j=1\dots n-1$. Finally, we have a \emph{special man} $m^*$ and \emph{special woman} $w^*$. This special pair is the one we will test for stability. Let the preferences be

\begin{align*}
  m_u: & \{w_v:\langle u, v\rangle \geq l\} \succ \{w_j\}_{j=1}^{n-1} \succ w^* \succ \{w_v:\langle u, v\rangle < l\} & \forall u\in U\\
  m_i: & \{w_v\} \succ \{w_j\}_{j=1}^{n-1} \succ w^* & \forall i\in\{1\dots n-1\}\\
  m^*: & w^* \succ \{w_v\} \succ \{w_j\}_{j=1}^{n-1}\\
  w_v: & \{m_u:\langle u, v\rangle \geq l\} \succ \{m_i\}_{i=1}^{n-1} \succ m^* \succ \{m_u:\langle u, v\rangle < l\} & \forall v\in V\\
  w_j: & \{m_u\} \succ \{m_i\}_{i=1}^{n-1} \succ m^* & \forall j\in\{1\dots n-1\}\\
  w^*: & \{m_i\}_{i=1}^{n-1} \succ \{m_u\} \succ m^*
\end{align*}

so that, for example, man $m_u$ corresponding to $u\in U$ will most prefer women $w_v$ for some $v\in V$ with $\langle u,v\rangle\geq l$ (in decreasing order of $\langle u,v\rangle$), then all of the dummy women (equally), then the special woman $w^*$, and finally the remaining women $w_v$ (in decreasing order of $\langle u,v\rangle$).

  First suppose for some $\hat{u}\in U$ and $\hat{v}\in V$ we have $\langle \hat{u},\hat{v} \rangle \geq l$ and let this be the pair with largest inner product. Now consider the deferred acceptance algorithm for finding the woman-optimal stable matching. First, $w_{\hat{v}}$ will propose to $m_{\hat{u}}$ and will be accepted. The dummy women will propose to the remaining men corresponding to $U$. Then any other woman $w_v$ will be accepted by either a dummy man or a man $m_u$, causing the dummy woman matched with him to move to a dummy man. In any case, all men besides $m^*$ are matched to a woman they prefer over $w^*$, so when she proposes to them, they will reject her. Thus $w^*$ will match with $m^*$. Since $w^*$ receives her least preferred choice in the woman optimal stable matching, $(m^*,w^*)$ is a pair in every stable matching.

Now suppose $\langle u,v \rangle < l$ for every $u\in U$,$v\in V$. Consider the deferred acceptance algorithm for finding the man-optimal stable matching. First, the dummy men will propose to the women corresponding to $V$ and will be accepted. Then every man $m_u$ will propose to the dummy women, but only $n-1$ of them can be accepted. The remaining one will propose to $w^*$. When $m^*$ proposes to $w^*$, she rejects him, causing him to eventually be accepted by the available woman $w_v$. Thus $m^*$ will not match with $w^*$ in any stable matching since she is his most preferred choice but he is not matched with her in the man-optimal stable matching, so $(m^*,w^*)$ is not a pair in any stable matching. Figure \ref{fig:stablepair} demonstrates each of these cases.

Since the stable pair questions for whether $(m^*,w^*)$ are a stable pair in any or all stable matchings are equivalent with these preferences, this reduction works for both.

Finally, we claim the following vectors realize the preferences above for the attribute model. We leave it to the reader to verify this. As in our other hardness reductions, the weight and attribute vectors are identical for each participant.

\[
\begin{array}{r l c c c c c c c c c c}
  m_u: & u^7   & \circ & 1^{7(l-1)} & \circ & 0^{7(l-1)} & \circ & 1^6            & \circ & 0^6            & \circ & 0^6\\
  m_i: & 0^{7d}& \circ & 1^{7(l-1)} & \circ & 1^{7(l-1)} & \circ & 0^6            & \circ & 1^6            & \circ & 0^6\\
  m^*: & 0^{7d}& \circ & 1^{7(l-1)} & \circ & 1^{7(l-1)} & \circ & 0^6            & \circ & 0^6            & \circ & 1^6\\
\hline
  w_v: & v^7   & \circ & 0^{7(l-1)} & \circ & 1^{7(l-1)} & \circ & 0^6            & \circ & 1^6            & \circ & (1\circ 0^5)\\
  w_j: & 0^{7d}& \circ & 1^{7(l-1)} & \circ & 0^{7(l-1)} & \circ & 1^6            & \circ & (1^5\circ 0)   & \circ & 0^6\\
  w^*: & 0^{7d}& \circ & 1^{7(l-1)} & \circ & 0^{7(l-1)} & \circ & (1^3\circ 0^3) & \circ & (1^4\circ 0^2) & \circ & (1^2\circ 0^4)\\
\end{array}
\]


\end{proof}

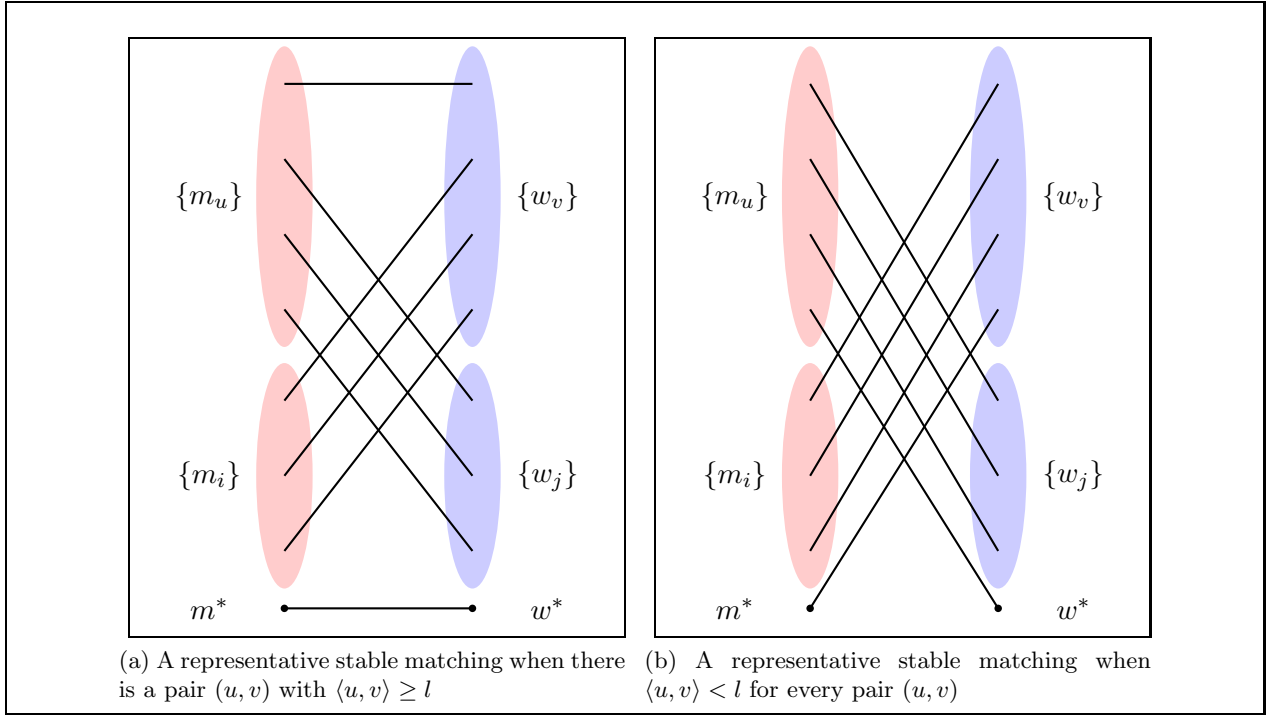
\begin{figure}[H]
\centering
\caption{A representation of the reduction from maximum inner product
  to checking a stable pair}
\label{fig:stablepair}
\subfloat[A representative stable matching when there is a pair $(u,v)$ with $\langle u,v \rangle \geq l$]{
\framebox[.4\textwidth]{
\begin{tikzpicture}[node distance=0.5cm]

    \path node[fill=red!20,ellipse,name=mu,minimum width=0.75cm,minimum height=4cm] at (0,0) {}
          node [fill=red!20,ellipse,name=mi,minimum width=0.75cm,minimum height=3cm,below = 0.2cm of mu] {}
          node [fill,circle,name=ms,radius=2pt,inner sep=1pt,below = 0.2cm of mi] {};
    \path node[list,left of=mu,xshift = -0.5cm] (lmu) {$\{m_u\}$}
          node[list,left of=mi,xshift = -0.5cm] (lmi) {$\{m_i\}$}
          node[list,left of=ms,xshift = -0.5cm] (lms) {$m^*$};
    \path node[fill=blue!20,ellipse,name=wv,right of=mu,xshift = 2cm,minimum width=0.75cm,minimum height=4cm] {}
          node [fill=blue!20,ellipse,name=wj,minimum width=0.75cm,minimum height=3cm,below = 0.2cm of wv] {}
          node [fill,circle,name=ws,radius=2pt,inner sep=1pt,below = 0.2cm of wj] {};
    \path node[list,right of=wv,xshift = 0.5cm] (lwv) {$\{w_v\}$}
          node[list,right of=wj,xshift = 0.5cm] (lwj) {$\{w_j\}$}
          node[list,right of=ws,xshift = 0.5cm] (lws) {$w^*$};
    \draw[thick] ([yshift = 1.5cm]mu.center) edge ([yshift = 1.5cm]wv.center);
    \draw[thick] ([yshift = 0.5cm]mu.center) edge ([yshift = 1cm]wj.center);
    \draw[thick] ([yshift = -0.5cm]mu.center) edge (wj.center);
    \draw[thick] ([yshift = -1.5cm]mu.center) edge ([yshift = -1cm]wj.center);
    \draw[thick] ([yshift = 1cm]mi.center) edge ([yshift = 0.5cm]wv.center);
    \draw[thick] (mi.center) edge ([yshift = -0.5cm]wv.center);
    \draw[thick] ([yshift = -1cm]mi.center) edge ([yshift = -1.5cm]wv.center);
    \draw[thick] (ms.center) edge (ws.center);
\end{tikzpicture}
}}
~
\subfloat[A representative stable matching when $\langle u,v \rangle < l$ for every pair $(u,v)$]{
\framebox[.4\textwidth]{
\centering
\begin{tikzpicture}[node distance=0.5cm]

    \path node[fill=red!20,ellipse,name=mu,minimum width=0.75cm,minimum height=4cm] at (0,0) {}
          node [fill=red!20,ellipse,name=mi,minimum width=0.75cm,minimum height=3cm,below = 0.2cm of mu] {}
          node [fill,circle,name=ms,radius=2pt,inner sep=1pt,below = 0.2cm of mi] {};
    \path node[list,left of=mu,xshift = -0.5cm] (lmu) {$\{m_u\}$}
          node[list,left of=mi,xshift = -0.5cm] (lmi) {$\{m_i\}$}
          node[list,left of=ms,xshift = -0.5cm] (lms) {$m^*$};
    \path node[fill=blue!20,ellipse,name=wv,right of=mu,xshift = 2cm,minimum width=0.75cm,minimum height=4cm] {}
          node [fill=blue!20,ellipse,name=wj,minimum width=0.75cm,minimum height=3cm,below = 0.2cm of wv] {}
          node [fill,circle,name=ws,radius=2pt,inner sep=1pt,below = 0.2cm of wj] {};
    \path node[list,right of=wv,xshift = 0.5cm] (lwv) {$\{w_v\}$}
          node[list,right of=wj,xshift = 0.5cm] (lwj) {$\{w_j\}$}
          node[list,right of=ws,xshift = 0.5cm] (lws) {$w^*$};
    \draw[thick] ([yshift = -1.5cm]mu.center) edge (ws.center);
    \draw[thick] ([yshift = 1.5cm]mu.center) edge ([yshift = 1cm]wj.center);
    \draw[thick] ([yshift = 0.5cm]mu.center) edge (wj.center);
    \draw[thick] ([yshift = -0.5cm]mu.center) edge ([yshift = -1cm]wj.center);
    \draw[thick] ([yshift = 1cm]mi.center) edge ([yshift = 1.5cm]wv.center);
    \draw[thick] (mi.center) edge ([yshift = 0.5cm]wv.center);
    \draw[thick] ([yshift = -1cm]mi.center) edge ([yshift = -0.5cm]wv.center);
    \draw[thick] (ms.center) edge ([yshift = -1.5cm]wv.center);
\end{tikzpicture}
}}
\end{figure}

This reduction also has consequences on the existence of nondeterministic algorithms for the stable pair problem assuming the \emph{Nondeterministic Strong Exponential Time Hypothesis}.

\begin{definition}[\cite{CarmosinoGaoImpagliazzoMikhailinPaturiSchneider2016}]
  The \emph{Nondeterministic Strong Exponential Time Hypothesis} ($\NSETH$) stipulates that
  for each $\varepsilon > 0$ there is a $k$ such that $\ksat$ requires
  co-nondeterministic time $\Omega(2^{(1-\varepsilon) n})$.
\end{definition}

In other words, the Nondeterministic Strong Exponential Time
Hypothesis stipulates that for $\cnfsat$ there is no proof of
unsatisfiability that can be checked deterministically in time
$\Omega(2^{(1-\varepsilon) n})$.

Assuming $\NSETH$, any problem that is $\SETH$-hard at time $T(n)$
under deterministic reductions either require $T(n)$ time
nondeterministically or co-nondeterministically, i.e. either there is
no proof that an instance is true or there is no proof that an
instance is false that can be checked in time faster than $T(n)$. Note
that all reductions in this paper are deterministic. In particular,
the maximum inner product problem does not have a
$O(N^{2-\varepsilon})$ co-nondeterministic time algorithm for any
$\varepsilon > 0$ assuming $\NSETH$, since it has a simple linear time
nondeterministic algorithm.

Since the reduction of Theorem \ref{thm:stablepair} is a simple
reduction that maps a true instance of maximum inner product to a true
instance of the stable pair problem, we can conclude that the stable
pair problem is also hard co-nondeterministically.

\begin{corollary}
  Assuming $\NSETH$, for any $\varepsilon > 0$, there is a $c$ such
  that determining whether a given pair is part of any or all stable
  matchings in the boolean $d$-attribute model with $d = c \log n$
  dimensions requires co-nondeterministic time
  $\Omega(n^{2-\varepsilon})$.
\end{corollary}

We also have a reduction so that the given pair is stable in any or
all stable matchings if and only if there is not a pair of vectors
with large inner product. This shows that the stable pair problem is
also hard nondeterministically.

\begin{theorem}
  \label{thm:stablepairother}
  Assuming $\NSETH$, for any $\varepsilon > 0$, there is a $c$ such
  that determining whether a given pair is part of any or all stable
  matchings in the boolean $d$-attribute model with $d = c \log n$
  dimensions requires nondeterministic time
  $\Omega(n^{2-\varepsilon})$.
\end{theorem}
\begin{proof}
  This reduction uses the same setup as the one in Theorem \ref{thm:stablepair} except that we now have $n$ dummy men and women instead of $n-1$ and we slightly change the preferences as follows:

\begin{align*}
  m_u: & \{w_v:\langle u, v\rangle \geq l\} \succ \{w_j\}_{j=1}^{n} \succ w^* \succ \{w_v:\langle u, v\rangle < l\} & \forall u\in U\\
  m_i: & \{w_v\} \succ \mathbf{w^* \succ \{w_j\}_{j=1}^{n}} & \forall i\in\{1\dots n\}\\
  m^*: & w^* \succ \{w_v\} \succ \{w_j\}_{j=1}^{n}\\
  w_v: & \{m_u:\langle u, v\rangle \geq l\} \succ \{m_i\}_{i=1}^{n} \succ m^* \succ \{m_u:\langle u, v\rangle < l\} & \forall v\in V\\
  w_j: & \{m_u\} \succ \{m_i\}_{i=1}^{n} \succ m^* & \forall j\in\{1\dots n\}\\
  w^*: & \{m_i\}_{i=1}^{n} \succ \{m_u\} \succ m^*
\end{align*}

  First suppose for some $\hat{u}\in U$ and $\hat{v}\in V$ we have $\langle \hat{u},\hat{v} \rangle \geq l$ and let this be the pair with largest inner product. Consider the deferred acceptance algorithm for finding the man-optimal stable matching. First, some of the men corresponding to $U$ will propose to the women corresponding to $V$ and at least $m_{\hat{u}}$ will be accepted by $w_{\hat{v}}$. The remaining men corresponding to $U$ will be accepted by dummy women. The dummy men will propose to the women corresponding to $V$ but not all can be accepted. These rejected dummy men will propose to $w^*$ who will accept one. Then when $m^*$ proposes to $w^*$ she will reject him, as will the women corresponding to $V$, so he will be matched with a dummy woman. Since $m^*$ and $w^*$ are not matched in the man optimal stable matching, $(m^*,w^*)$ is not a pair in any stable matching.

Now suppose $\langle u,v \rangle < l$ for every $u\in U$,$v\in V$ and consider the deferred acceptance algorithm for finding the woman-optimal stable matching. First, the dummy women will propose to the men corresponding to $U$ and will be accepted. Then every woman $w_v$ will propose to the dummy men and be accepted. Since every man besides $m^*$ is matched with a woman he prefers to $w^*$, when she proposes to them, she will be rejected, so she will pair with $m^*$. Since $w^*$ receives her least preferred choice in the woman optimal stable matching, $(m^*,w^*)$ is a pair in every stable matching. Figure \ref{fig:costablepair} demonstrates each of these cases.

We can amend the vectors from Theorem \ref{thm:stablepair} as follows so that they realize the changed preferences with the attribute model.

\[
\begin{array}{r l c c c c c c c c c c}
  m_u: & u^7   & \circ & 1^{7(l-1)} & \circ & 0^{7(l-1)} & \circ & 1^6            & \circ & 0^6            & \circ & 0^6\\
  m_i: & 0^{7d}& \circ & 1^{7(l-1)} & \circ & 1^{7(l-1)} & \circ & 0^6            & \circ & 1^6            & \circ & 0^6\\
  m^*: & 0^{7d}& \circ & 1^{7(l-1)} & \circ & 1^{7(l-1)} & \circ & 0^6            & \circ & 0^6            & \circ & 1^6\\
\hline
  w_v: & v^7   & \circ & 0^{7(l-1)} & \circ & 1^{7(l-1)} & \circ & 0^6            & \circ & 1^6            & \circ & (1\circ 0^5)\\
  w_j: & 0^{7d}& \circ & 1^{7(l-1)} & \circ & 0^{7(l-1)} & \circ & 1^6            & \circ & \mathbf{(1^3\circ 0^3)}   & \circ & 0^6\\
  w^*: & 0^{7d}& \circ & 1^{7(l-1)} & \circ & 0^{7(l-1)} & \circ & (1^3\circ 0^3) & \circ & (1^4\circ 0^2) & \circ & (1^2\circ 0^4)\\
\end{array}
\]


\end{proof}

\begin{figure}[H]
\centering
\caption{A representation of the reduction from maximum inner product
  to checking a stable pair such that a true maximum inner product
  instance maps to a false stable pair instance}
\label{fig:costablepair}
\subfloat[A representative stable matching when there is a pair $(u,v)$ with $\langle u,v \rangle \geq l$]{
\framebox[.4\textwidth]{
\centering
\begin{tikzpicture}[node distance=0.5cm]

    \path node[fill=red!20,ellipse,name=mu,minimum width=0.75cm,minimum height=4cm] at (0,0) {}
          node [fill=red!20,ellipse,name=mi,minimum width=0.75cm,minimum height=4cm,below = 0.2cm of mu] {}
          node [fill,circle,name=ms,radius=2pt,inner sep=1pt,below = 0.2cm of mi] {};
    \path node[list,left of=mu,xshift = -0.5cm] (lmu) {$\{m_u\}$}
          node[list,left of=mi,xshift = -0.5cm] (lmi) {$\{m_i\}$}
          node[list,left of=ms,xshift = -0.5cm] (lms) {$m^*$};
    \path node[fill=blue!20,ellipse,name=wv,right of=mu,xshift = 2cm,minimum width=0.75cm,minimum height=4cm] {}
          node [fill=blue!20,ellipse,name=wj,minimum width=0.75cm,minimum height=4cm,below = 0.2cm of wv] {}
          node [fill,circle,name=ws,radius=2pt,inner sep=1pt,below = 0.2cm of wj] {};
    \path node[list,right of=wv,xshift = 0.5cm] (lwv) {$\{w_v\}$}
          node[list,right of=wj,xshift = 0.5cm] (lwj) {$\{w_j\}$}
          node[list,right of=ws,xshift = 0.5cm] (lws) {$w^*$};
    \draw[thick] ([yshift = 1.5cm]mu.center) edge ([yshift = 1.5cm]wv.center);
    \draw[thick] ([yshift = 0.5cm]mu.center) edge ([yshift = 1.5cm]wj.center);
    \draw[thick] ([yshift = -0.5cm]mu.center) edge ([yshift = 0.5cm]wj.center);
    \draw[thick] ([yshift = -1.5cm]mu.center) edge ([yshift = -0.5cm]wj.center);
    \draw[thick] ([yshift = 1.5cm]mi.center) edge ([yshift = 0.5cm]wv.center);
    \draw[thick] ([yshift = 0.5cm]mi.center) edge ([yshift = -0.5cm]wv.center);
    \draw[thick] ([yshift = -0.5cm]mi.center) edge ([yshift = -1.5cm]wv.center);
    \draw[thick] ([yshift = -1.5cm]mi.center) edge (ws.center);
    \draw[thick] (ms.center) edge ([yshift = -1.5cm]wj.center);
\end{tikzpicture}
}}
~
\subfloat[A representative stable matching when $\langle u,v \rangle < l$ for every pair $(u,v)$]{
\framebox[.4\textwidth]{
\centering
\begin{tikzpicture}[node distance=0.5cm]

    \path node[fill=red!20,ellipse,name=mu,minimum width=0.75cm,minimum height=4cm] at (0,0) {}
          node [fill=red!20,ellipse,name=mi,minimum width=0.75cm,minimum height=4cm,below = 0.2cm of mu] {}
          node [fill,circle,name=ms,radius=2pt,inner sep=1pt,below = 0.2cm of mi] {};
    \path node[list,left of=mu,xshift = -0.5cm] (lmu) {$\{m_u\}$}
          node[list,left of=mi,xshift = -0.5cm] (lmi) {$\{m_i\}$}
          node[list,left of=ms,xshift = -0.5cm] (lms) {$m^*$};
    \path node[fill=blue!20,ellipse,name=wv,right of=mu,xshift = 2cm,minimum width=0.75cm,minimum height=4cm] {}
          node [fill=blue!20,ellipse,name=wj,minimum width=0.75cm,minimum height=4cm,below = 0.2cm of wv] {}
          node [fill,circle,name=ws,radius=2pt,inner sep=1pt,below = 0.2cm of wj] {};
    \path node[list,right of=wv,xshift = 0.5cm] (lwv) {$\{w_v\}$}
          node[list,right of=wj,xshift = 0.5cm] (lwj) {$\{w_j\}$}
          node[list,right of=ws,xshift = 0.5cm] (lws) {$w^*$};
    \draw[thick] ([yshift = 1.5cm]mu.center) edge ([yshift = 1.5cm]wj.center);
    \draw[thick] ([yshift = 0.5cm]mu.center) edge ([yshift = 0.5cm]wj.center);
    \draw[thick] ([yshift = -0.5cm]mu.center) edge ([yshift = -0.5cm]wj.center);
    \draw[thick] ([yshift = -1.5cm]mu.center) edge ([yshift = -1.5cm]wj.center);
    \draw[thick] ([yshift = 1.5cm]mi.center) edge ([yshift = 1.5cm]wv.center);
    \draw[thick] ([yshift = 0.5cm]mi.center) edge ([yshift = 0.5cm]wv.center);
    \draw[thick] ([yshift = -0.5cm]mi.center) edge ([yshift = -0.5cm]wv.center);
    \draw[thick] ([yshift = -1.5cm]mi.center) edge ([yshift = -1.5cm]wv.center);
    \draw[thick] (ms.center) edge (ws.center);
\end{tikzpicture}
}}
\end{figure}
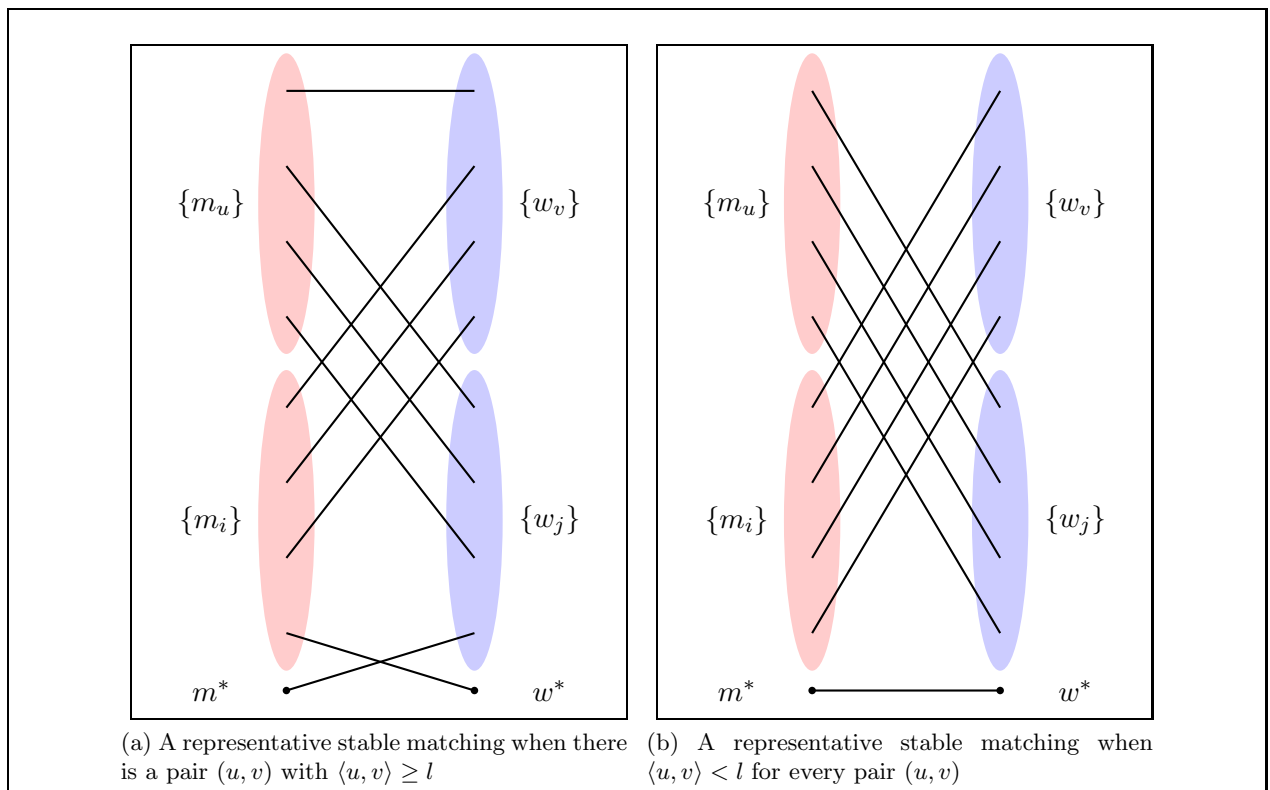

We would like to point out that the results on the hardness for
(co-)nondeterministic algorithms do not apply to Merlin-Arthur ($\MA$)
algorithms, i.e. algorithms with access to both nondeterministic bits
and randomness. Williams \cite{Williams2016} gives fast $\MA$ algorithms for a
number of $\SETH$-hard problems, and the same techniques also yield a
$O(dn)$ time $\MA$ algorithm for the verification of a stable matching
in the boolean attribute model with $d$ attributes. We can obtain
$\MA$ algorithms with time $O(dn)$ for finding stable matchings and
certifying that a pair is in at least one stable matching by first
nondeterministically guessing a stable matching.

\section{Other Succinct Preference Models}
\label{sec:otherpref}
In this section, we provide subquadratic algorithms for other succinct preference models, single-peaked and geometric, which are motivated by economics.

\subsection{One Dimensional Single-Peaked Preferences}
Formally, we say the men's preferences over the women in a matching market are \emph{single-peaked} if the women can be ordered as points along a line ($p(w_1) < p(w_2) < \dots < p(w_n)$) and for each man $m$ there is a point $q(m)$ and a binary preference relation $\succ_m$ such that if $p(w_i) 
\leq q(m)$ then $p(w_i) \succ_m p(w_j)$ for $j < i$ and if $p(w_i) \geq q(m)$ then $p(w_i) \succ_m p(w_j)$ for $j > i$. Essentially, each man prefers the women that are ``closest'' to his ideal point $q(m)$. One example of a preference relation for $m$ would be the distance from $q(m)$. If the women's preferences are also single-peaked then we say the matching market has single peaked preferences. Since these preferences only consist of the $p$ and $q$ values and the preference relations for the participants, they can be represented succinctly as long as the relations require subquadratic space.

\subsubsection{Verifying a Stable Matching for Single-Peaked Preferences}

Here we demonstrate a subquadratic algorithm for verifying if a given matching is stable when the preferences of the matching market are single-peaked. We assume that the preference relations can be computed in constant time.

\begin{theorem}
  There is an algorithm to verify a stable matching in the single-peaked
  preference model in $O(n\log n)$ time.
\end{theorem}

\begin{proof}
  Let $p(m_i)$ be the point associated with man $m_i$, $q(m_i)$ be $m_i$'s preference point, and $\succ_{m_i}$ be $m_i$'s preference relation. The women's points are denoted analogously. We assume that $p(m_i) < p(m_j)$ if and only if $i<j$ and the same for the women. Let $\mu$ be the given matching we are to check for stability.

First, for each man $m$, we compute the intervals along the line of women which includes all women $m$ strictly prefers to $\mu(m)$. If this interval is empty, $m$ is with his most preferred woman and cannot be involved in any blocking pairs so we can ignore him. For all nonempty intervals each endpoint is $p(w)$ for some woman $w$. We also compute these intervals for the women. Note that for any man $m$ and woman $w$, $(m,w)$ is a blocking pair for $\mu$ if and only if $m$ is in $w$'s interval and $w$ is in $m$'s interval.

We will process each of the women in order from $w_1$ to $w_n$ maintaining a balanced binary search tree of the men who prefer that woman to their partners. This will allow us to easily check if she prefers any of them by seeing if any elements in the tree are between the endpoints of her interval. Initially this tree is empty. When processing a woman $w$, we first add any man $m$ whose interval begins with $w$ to the search tree. Then we check to see if $w$ prefers any men in the tree. If so, we know the matching is not stable. Otherwise, we remove any man $m$ from the tree whose interval ends with $w$ and proceed to the next woman. Algorithm \ref{alg:singlepeakverification} provides pseudocode for this algorithm.

Computing the intervals requires $O(n\log n)$. Since we only insert each man into the tree at most once, maintaining the tree requires $O(n\log n)$. The queries also require $O(\log n)$ for each woman so the total time is $O(n\log n)$.
\end{proof}

\begin{algorithm}[H]
\For{each woman $w$}{
   Create two empty lists $w.begin$ and $w.end$.\\
   Use binary search to find the leftmost man $m$ and rightmost man $m'$ that $w$ prefers to $\mu(w)$ if any. (Otherwise remove $w$.)\\
   Let $w.s = p(m)$ and $w.t = p(m')$.
}
\For{each man $m$}{
	Use binary search to find the leftmost woman $w$ and rightmost woman $w'$ that $m$ prefers to $\mu(m)$ if any. (Otherwise ignore $m$.)\\
   Add $m$ to $w.begin$ and $w'.end$.
}
Initialize an empty balanced binary search tree $T$.\\
\For{$i=1$ to $n$}{
   \For{$m\in w_i.begin$}{
      $T.\insrt(p(m))$
   }
   \If{there are any points $p(m)$ in $T$ between $w_i.s$ and $w_i.t$}{
      \Return{$(m,w_i)$ is a blocking pair.}
   }
   \For{$m\in w_i.end$}{
      $T.\delete(p(m))$
	}
}
\Return{$\mu$ is stable.}
\caption{Single-Peaked Stable Matching Verification}
\label{alg:singlepeakverification}
\end{algorithm}

\subsubsection{Remarks on Finding a Stable Matching for Single-Peaked Preferences}

The algorithm in \cite{BartholdiTrick1986} relies on the observation that there will always be a pair or participants who are each other's first choice with narcissistic single-peaked preferences. Thus a greedy approach where one such pair is selected and then removed works well. However, this is not the case when we remove the narcissistic assumption. In fact, as with the two-list case, Table \ref{tab:singlepeakpreferences} presents an example where no participant is matched with their top choice in the unique stable matching. Note that the preferences for the men and women are symmetric. The reader can verify that these preferences can be realized in the single-peaked preference model using the orderings $p(m_1) < p(m_2) < p(m_3) < p(m_4)$ and $p(w_1) < p(w_2) < p(w_3) < p(w_4)$ and that the unique stable matching is $\{(m_1,w_4),(m_2,w_2),(m_3,w_3),(m_4,w_1)\}$ where no participant receives their first choice.

\begin{table}[H]
  \caption{Single-peaked preferences where no participant receives their top choice in the stable matching}
  \label{tab:singlepeakpreferences}
  \begin{center}
    \begin{tabular}{c|c}
      Man & Preference List\\
      \hline
      $m_1$ & $w_3\succ w_2\succ w_4\succ w_1$\\
      $m_2$ & $w_3\succ w_2\succ w_4\succ w_1$\\
      $m_3$ & $w_4\succ w_3\succ w_2\succ w_1$\\
      $m_4$ & $w_2\succ w_1\succ w_3\succ w_4$\\
    \end{tabular}
    \quad
    \begin{tabular}{c|c}
      Woman & Preference List\\
      \hline
      $w_1$ & $m_3\succ m_2\succ m_4\succ m_1$\\
      $w_2$ & $m_3\succ m_2\succ m_4\succ m_1$\\
      $w_3$ & $m_4\succ m_3\succ m_2\succ m_1$\\
      $w_4$ & $m_2\succ m_1\succ m_3\succ m_4$\\
    \end{tabular}
  \end{center}
\end{table}

Also no greedy algorithm following the model inspired by
\cite{DavisImpagliazzo2009} will succeed for single-peak preferences
because the preferences in Table \ref{tab:greedylist} can be realized
in the single-peaked preference model using the orderings
$p(m_1) < p(m_2) < p(m_3)$ and $p(w_1) < p(w_2) < p(w_3)$.

\subsection{Geometric Preferences}
We say the men's preferences over the women in a matching market are
\emph{geometric} in $d$ dimensions if each women $w$ is defined by a
\emph{location} $p(w)$ and for each man $m$ there is an \emph{ideal}
$q(m)$ such that $m$ prefers woman $w_1$ to $w_2$ if and only if
$\|p(m) - q(w_1)\|_2^2 < \|p(m) - q(w_2)\|_2^2$, i.e. $p(w_1)$ has smaller
euclidean distance from the man's ideal than $p(w_2)$. If the women's
preferences are also geometric we call the matching market
geometric. We further call the preferences \emph{narcissistic} if
$p(x) = q(x)$ for every participant $x$. Our results for the attribute
model extend to geometric preferences. 

Note that one-dimensional geometric preferences are a special case of
single-peaked preferences. As such, geometric preferences might be used to model preferences over
political candidates who are given a score on several (linear) policy
areas, e.g. protectionist vs. free trade and hawkish vs. dovish foreign
policy.

Arkin et al.~\cite{ArkinBaeEfrat2009} also consider geometric
preferences, but restrict themselves to the narcissistic case. Our
algorithms do not require the preferences to be narcissistic, hence
our model is more general. On the other hand, our lower bounds for
large dimensions also apply to the narcissistic special case. While
Arkin et al.~take special care of different notions of stability in
the presence of ties, we concentrate on weakly stable matchings.  Although we restrict ourselves to the stable matching problem for the
sake of presentation, all lower bounds and verification algorithms
naturally extend to the stable roommate problem. Since all proofs in this section are closely related those for the attribute
model, we restrict ourselves to proof sketches highlighting the main
differences.

Theorem \ref{thm:smallstablematching} extends immediately to the
geometric case without any changes in the proof.

\begin{corollary}[Geometric version of Theorem \ref{thm:smallstablematching}]
\label{cor:smallstablematchinggeom}
  There is an algorithm to find a stable matching in the
  $d$-dimensional geometric model with at most a constant $C$ distinct 
  values in time $O(C^{2d}n(d + \log n))$.
\end{corollary}

For the verification of a stable matching with real-valued vectors we
use a standard lifting argument.

\begin{corollary}[Geometric version of Theorem \ref{thm:verifyreal}]
\label{cor:verifyrealgeom}
  There is an algorithm to verify a stable matching in the
  $d$-dimensional geometric model with real-valued locations and ideals in time
  $\tilde{O}(n^{2-1/2(d+1)})$
\end{corollary}
\begin{proof}
  Let $q \in \reals^d$ be an ideal and let $a, b \in \reals^d$ be two locations. Define
  $q', a', b' \in \reals^{d+1}$ as $q' = (q_1, \ldots, q_d, -1/2)$,
  $a' = (a_1, \ldots, a_d, \sum_{i=1}^d a_i^2)$ and $b' = (b_1, \ldots,
  b_d, \sum_{i=1}^d b_i^2)$.

  We have $\langle a', q \rangle = 1/2 \sum_{i=1}^d q_i - 1/2
  \|q-a\|_2^2$. Hence we get $\|q-a\|_2^2 < \|q-b\|_2^2$ if and only if
  $\langle q', a'\rangle > \langle q', b'\rangle$, so we can reduce
  the stable matching problem in the $d$-dimensional geometric model
  to the $d+1$-attribute model.
\end{proof}

For the boolean case, we can adjust the proof of
Theorem \ref{thm:verifyboolean} by using a threshold of parities
instead of a threshold of conjunctions. The degree of the resulting
polynomial remains the same.

\begin{corollary}[Geometric version of Theorem \ref{thm:verifyboolean}]
  \label{thm:verifybooleangeom}
  In the geometric model with $n$ men and women, with locations and
  ideals in $\{0,1\}^d$ with $d = c \log n$, there is a randomized
  algorithm to decide if a given matching is stable in time
  $\tilde{O}(n^{2-1/O(c \log^2(c))})$ with error probability at most
  $1/3$.
\end{corollary}

For lower bounds we reduce from the minimum Hamming distance problem
which is $\SETH$-hard with the same parameters as the maximum inner
product problem \cite{AlmanWilliams2015}. The Hamming distance of two
boolean vectors is exactly their squared euclidean distance, hence a
matching market where the preferences are defined by Hamming distances
is geometric.

\begin{definition}
  For any $d$ and input $l$, the minimum Hamming distance problem
  is to decide if two input sets $U, V
  \subseteq \{0,1\}^d$ with $|U| = |V| = n$ have a pair $u \in U$, $v
  \in V$ such that $\|u - v\|_2^2 < l$.
\end{definition}

\begin{lemma}[\cite{AlmanWilliams2015}]
  Assuming $\SETH$, for any $\varepsilon > 0$, there is a $c$ such
  that solving the minimum Hamming distance problem on $d = c \log n$
  dimensions requires time $\Omega(n^{2-\varepsilon})$.
\end{lemma}

For the hardness of finding a stable matching, the construction from
Theorem \ref{thm:hardnessfinding} works without adjustments.

\begin{corollary}[Geometric version of Theorem \ref{thm:hardnessfinding}]
  \label{thm:hardnessfindinggeom}
  Assuming $\SETH$, for any $\varepsilon > 0$, there is a $c$ such
  that finding a stable matching in the (boolean) $d$-dimensional
  geometric model with $d = c \log n$ dimensions requires time
  $\Omega(n^{2-\varepsilon})$.
\end{corollary}

For the hardness of verifying a stable matching, the construction is
as follows.

\begin{corollary}[Geometric version of Theorem \ref{thm:hardnessverify}]
  \label{thm:hardnessverifygeom}
  Assuming $\SETH$, for any $\varepsilon > 0$, there is a $c$ such
  that verifying a stable matching in the (boolean) $d$-dimensional
  geometric model with $d = c \log n$ dimensions requires time
  $\Omega(n^{2-\varepsilon})$.
\end{corollary}
\begin{proof}
  Let $U, V \subseteq \{0,1\}^d$ be the inputs to the minimum Hamming
  distance problem and let $l$ be the threshold.

  For every $u \in U$, define a real man $m_u$ with both ideal and
  location as $u \circ 0^l$ and a dummy woman $w'_u$ with ideal and
  location $u \circ 1^l$. Symmetrically for $v \in V$ define $w_v$
  with $v \circ 0^l$ and $m'_v$ with $v \circ 1^l$. The matching
  $(m_u, w'_u)$ for all $u \in U$ and $(w_v, m'_v)$ for all $v \in V$
  is stable if and only if there is there is no pair $u,v$ with
  Hamming distance less than $l$.
\end{proof}

The hardness results for checking a stable pair also
translate to the geometric model. In particular, since both variants
of the proof extend to the geometric model we have the same
consequences for nondeterministic algorithms.

\begin{corollary}[Geometric version of Theorem \ref{thm:stablepair}]
  \label{thm:stablepairgeom}
  Assuming $\SETH$, for any $\varepsilon > 0$, there is a $c$ such
  that determining whether a given pair is part of any or all stable
  matchings in the (boolean) $d$-dimensional geometric model with
  $d = c \log n$ dimensions requires time $\Omega(n^{2-\varepsilon})$.
\end{corollary}
\begin{proof}
  We again reduce from the minimum Hamming distance problem. We assume
  without loss of generality that $d$ is even and the threshold $l$ is
  exactly $d/2 + 1$, i.e. the instance is true if and only if there
  are vectors $u,v$ with Hamming distance at most $d/2$. We can reduce
  to this case from any other threshold by padding the vectors.

  We use the same preference orders as in the $d$-attribute model. The
  following narcissistic instance realizes the preference order from
  Theorem \ref{thm:stablepair}. For a vector $u \in \{0,1\}^d$,
  $\overline{u}$ denotes its component-wise complement.

  \begin{alignat*}{2}
    m_u: & (u \circ \overline{u} \circ u \circ \overline{u})^3 & \circ 000000000\\
    m_i: & 0^{12d} & \circ 100000000\\
    m^*: & 0^{12d} & \circ 001111111\\
    w_v: & (v \circ \overline{v} \circ v \circ \overline{v})^3 & \circ 000000000\\
    w_j: & (0^{2d} \circ 1^{2d})^3 & \circ 010000000\\
    w^*: & (0^{2d} \circ 1^{2d})^3 & \circ 101110000\\
  \end{alignat*}

  Likewise the preference orders for Theorem \ref{thm:stablepairother}
  are achieved by the following vectors.

  \begin{alignat*}{2}
    m_u: & (u \circ \overline{u} \circ u \circ \overline{u})^3 & \circ 000000000\mathbf{00}\\
    m_i: & 0^{12d} & \circ 100000000\mathbf{00}\\
    m^*: & 0^{12d} & \circ 001111111\mathbf{00}\\
    w_v: & (v \circ \overline{v} \circ v \circ \overline{v})^3 & \circ 000000000\mathbf{00}\\
    w_j: & (0^{2d} \circ 1^{2d})^3 & \circ 010000000\mathbf{11}\\
    w^*: & (0^{2d} \circ 1^{2d})^3 & \circ 101110000\mathbf{00}\\
  \end{alignat*}
\end{proof}

\subsection{Strategic Behavior}

With geometric and single-peaked preferences, we assume that the participants are not allowed to misrepresent their location points. Rather they may only misrepresent their preference ideal. As such, the results of this section do not apply when preferences are narcissistic.

\begin{theorem}\label{thm:geometricstrategy}
  There is no strategy proof algorithm to find a stable matching in the geometric preference model.
\end{theorem}

\begin{proof}
We consider one-dimensional geometric preferences. Let the preference points and ideals be as given in Table \ref{tab:geometricstrategy} which yield the depicted preference lists. As in the proof for \ref{thm:liststrategy}, there are two stable matchings: the man-optimal matching $\{(m_1,w_3),(m_2,w_1),(m_3,w_2)\}$ and the woman-optimal matching $\{(m_1,w_3),(m_2,w_2),(m_3,w_1)\}$. However, if $w_2$ changes her ideal to $5/3$ then her preference list is $m_2\succ m_1\succ m_3$. Now there is a unique stable matching which is $\{(m_1,w_3),(m_2,w_2),(m_3,w_1)\}$, the woman-optimal stable matching from the original preferences. Therefore, any mechanism that does not always output the woman optimal stable matching can be manipulated by the women to their advantage. Similarly, any mechanism that does not always output the man-optimal matching could be manipulated by the men in some instances. Thus there is no strategy-proof mechanism for geometric preferences.
\end{proof}

\begin{table}[H]
  \caption{Geometric preferences that can be manipulated}
  \label{tab:geometricstrategy}
  \begin{center}
    \begin{tabular}{c|c|c}
      Man & Location $(p)$ & Ideal $(q)$\\
      \hline
      $m_1$ & 1 & $7/3$\\
      $m_2$ & 2 & 1\\
      $m_3$ & 3 & $5/3$\\
    \end{tabular}
    \quad
    \begin{tabular}{c|c|c}
      Woman & Location $(p)$ & Ideal $(q)$\\
      \hline
      $w_1$ & 1 & 3\\
      $w_2$ & 2 & $7/3$\\
      $w_3$ & 3 & 3\\
    \end{tabular}
  \end{center}
  
  \begin{center}
    \begin{tabular}{c|c}
      Man & Preference List\\
      \hline
      $m_1$ & $w_2\succ w_3\succ w_1$\\
      $m_2$ & $w_1\succ w_2\succ w_3$\\
      $m_3$ & $w_2\succ w_1\succ w_3$\\
    \end{tabular}
    \quad
    \begin{tabular}{c|c}
      Woman & Preference List\\
      \hline
      $w_1$ & $m_3\succ m_2\succ m_1$\\
      $w_2$ & $m_2\succ m_3\succ m_1$\\
      $w_3$ & $m_3\succ m_2\succ m_1$\\
    \end{tabular}
  \end{center}
\end{table}

Since one-dimensional geometric preferences are a special case of single-peaked preferences the following corollary results directly from Theorem \ref{thm:geometricstrategy}.

\begin{corollary}
  There is no strategy proof algorithm to find a stable matching in the single-peaked preference model.
\end{corollary}

\section{Conclusion and Open Problems}
We give subquadratic algorithms for finding and verifying stable
matchings in the $d$-attribute model and $d$-list model. We also show
that, assuming $\SETH$, one can only hope to find such algorithms if
the number of attributes $d$ is bounded by $O(\log n)$.

For a number of cases there is a gap between the conditional lower
bound and the upper bound. Our algorithms with real attributes and
weights are only subquadratic if the dimension is constant. Even for small
constants our algorithm to find a stable matching is not tight, as it
is not subquadratic for any $d = O(\log n)$. The techniques we use
when the attributes and weights are small constants do not readily
apply to the more general case.

There is also a gap between the time complexity of our algorithms for finding
a stable matching and verifying a stable matching. It would
be interesting to either close or explain this gap. On the one hand, subquadratic algorithms for finding a stable matching would demonstrate that the attribute and list models are computationally simpler than the general preference model. On the other hand, proving that there are no subquadratic algorithms would show a distinction between the problems of finding and verifying a stable matching in these settings which does not exist for the general preference model. Currently, we do not have a subquadratic
algorithm for finding a stable matching even in the $2$-list case, while we have an optimal algorithm for verifying a stable matching for $d$ lists. This $2$-list case seems to be a good starting place for
further research.

Additionally it is worth considering succinct preference models for other computational problems that involve preferences to see if we can also develop improved algorithms for these problems. For example, the Top Trading Cycles algorithm \cite{ShapleyScarf1974} can be made to run in subquadratic time for $d$-attribute preferences (when $d$ is constant) using the ray shooting techniques applied in this paper to find participants' top choices.

\noindent
{\bf Acknowledgment:} We would like to thank Russell Impagliazzo, Vijay Vazirani, and the anonymous reviewers for helpful discussions and comments.

\bibliographystyle{plain}
\bibliography{refs}

\end{document}